\documentclass[a4paper,10pt]{amsart}
\usepackage[utf8]{inputenc}
\usepackage{algorithmic} 
\usepackage{url}
\usepackage{amsmath,amssymb,amsthm}
\usepackage{graphicx}
\usepackage{fancyhdr}
\usepackage{xcolor}
\usepackage{hyperref}
\usepackage{multirow}
\usepackage{multicol}

\newtheorem{theorem}{Theorem}

\newtheorem{lemma}{Lemma}
\newtheorem{remark}{Remark}

\begin{document}

\title[]{An iterative warping and clustering algorithm to estimate multiple wave-shape functions from a nonstationary oscillatory signal}
\author{Marcelo A. Colominas and Hau-Tieng Wu}
\date{\today}

\maketitle

\begin{abstract}
Nonsinusoidal oscillatory signals are everywhere. In practice, the nonsinusoidal oscillatory pattern, modeled as a 1-periodic wave-shape function (WSF), might vary from cycle to cycle. When there are finite different WSFs, $s_1,\ldots,s_K$, so that the WSF jumps from one to another suddenly, the different WSFs and jumps encode useful information.
We present an iterative warping and clustering algorithm to estimate $s_1,\ldots,s_K$ from a nonstationary oscillatory signal with time-varying amplitude and frequency, and hence the change points of the WSFs. The algorithm is a novel combination of time-frequency analysis, singular value decomposition entropy and vector spectral clustering.
We demonstrate the efficiency of the proposed algorithm with simulated and real signals, including the voice signal, arterial blood pressure, electrocardiogram and accelerometer signal.
Moreover, we provide a mathematical justification of the algorithm under the assumption that the amplitude and frequency of the signal are slowly time-varying and there are finite change points that model sudden changes from one wave-shape function to another one.
\end{abstract}

\section{Introduction}

Nonstationary oscillatory signals are everywhere and have attracted significant attention in the past decades. In addition to the time-varying frequency and amplitude, in practice the analysis is further challenged by the fact that many oscillatory signals do not oscillate sinusoidally; for example, sound (Fig. \ref{fig:Voz}), arterial blood pressure (ABP) (Fig. \ref{fig:ABP}), and electrocardiogram (ECG) (Fig. \ref{fig:ECG}). In plain English, such oscillatory signal repeats a {\em non-sinusoidal} pattern over time. Take ECG as an example. Each cycle in ECG represents one heart beat, and its morphology is far away from the usual sine wave. To model this kind of signal with non-sinusoidal oscillatory pattern, in \cite{HTWu2013}, the {\em wave-shape function} (WSF), which mathematically is a $1$-periodic function, is introduced (this model is abbreviated as \texttt{WSFv1} hereafter).
In the past decade, this non-sinusoidal oscillation phenomenon and the WSF model have attracted more and more attention \cite{tavallali2014extraction,hou2016extracting,Xu2018recursive,Yang2019multiresolution}, and several variations of the original WSF model have been proposed to capture finer structures hidden inside the WSF; for example, the sparse time-frequency representation (TFR) \cite{tavallali2014extraction}, a generalization of WSF (this model is abbreviated as \texttt{WSFv2} hereafter) to capture time-varying oscillatory pattern \cite{lin2016waveshape} and the wave-shape oscillatory model \cite{YTLin2019wave}. Note that the time-varying oscillatory pattern is common in biomedical signals. For example, the QT variability \cite{Berger2003} in ECG leads to the morphology change from one cycle to another.
Given such signals and models, there are many interesting and challenging signal processing missions. For example, how to \emph{decompose} a signal into its elementary modes (commonly understood as a \emph{decomposition mission}), or how to capture the non-sinusoidal patterns  (called a \emph{wave-shape analysis mission}).

Based on the above mentioned WSF model and its variations, various algorithms have been developed to handle these missions under various challenges; for example, the optimization algorithm based on the sparse TFR method extracts intrinsic frequency information and defines the notion of degree of nonlinearity \cite{tavallali2014extraction,hou2016extracting}, the de-shape algorithm based on \texttt{WSFv2} reduces the impact of harmonics associated with the nonsinusoidal oscillation \cite{lin2016waveshape}, the linear regression approach based on \texttt{WSFv1} recovers the WSF and it is useful to estimate the number of harmonics \cite{Wu2016modeling,ruiz2022wave}, and the nonlinear regression approach based on \texttt{WSFv2} decomposes signals into components with time-varying waveforms \cite{colominas2021decomposing}. There also exist optimization approaches to decompose the signals \cite{Xu2018recursive,Yang2019multiresolution}, and dynamic diffusion maps algorithm to recover the intrinsic dynamics \cite{YTLin2019wave}.
We shall mention that while not motivated and developed based on the idea of WSF, the periodicity transform \cite{sethares1999,PPV3} and its generalization to time-frequency domain \cite{RDS2000} are other approaches that have been developed to handle similar non-sinusoidal oscillatory signals.
Associated developed algorithms have also been applied to various problems, for example, haemodynamic waveform analysis \cite{Pahlevan_Tavallali_Rinderknecht_Petrasek_Matthews_Hou_Gharib2014}, insulin resistance research \cite{petrasek2015intrinsic}, left ventricular ejection fraction study \cite{pahlevan2017noninvasive}, fetal ECG extraction from the trans-abdominal maternal ECG \cite{su2017extract}, photoplethysmography (PPG) analysis \cite{li2019non} and its application to the lower body negative pressure \cite{ShelleyAlianWu2022} and motion artifact detection \cite{maity2022ppgmotion}, accelerometer analysis for gait cadence \cite{wu2022application}, among others \cite{huang2021exploration}.

While there have been several efforts put in this research direction, however, to our knowledge, so far there is limited work focusing on handling the case when the WSF would {\em suddenly change} from one pattern to a dramatically different pattern, except a stretch of the wave-shape oscillatory model \cite{YTLin2019wave} and a simplified oscillatory change point detection model \cite{zhou2020frequency}. This sudden change phenomenon belongs to the wave-shape analysis mission, and it is commonly encountered in practice, in particular in the biomedical signal processing field. A typical example is the ECG signal when arrhythmia happens. See Fig. \ref{fig:ECG} for an illustration. In this example, the subject suffers from premature ventricular contraction (PVC). A visual inspection of Fig. \ref{fig:ECG} shows that there are two visually different oscillatory patterns (or WSFs); one pattern is associated with normal sinus rhythm, and the other one is associated with PVC, and the WSF could jump from one to another suddenly.

In this paper, we study this sudden WSF change problem, and focus on estimating all WSFs in a given signal. Specifically, suppose there are finite different WSFs ($1$-periodic functions), denoted as $s_1,\ldots,s_K$, that model different oscillatory patterns that happen in a signal, where we assume the oscillatory pattern may jump from one to another suddenly. A simplified example is
\begin{equation}\label{eq:WSF-special0}
x(t) = A\sum_{j=1}^{r+1} s_{l_j}(\xi t)\chi_{(t_{j-1},t_j]}(t)\,,
\end{equation}
where $r\geq 1$ is the number of change points of WSFs, $t_1<t_2<\ldots<t_r$ are the change points, $\chi$ is the indicator function, $A>0$ is the amplitude, $\xi>0$ is the frequency, $l_j\in \{1,\ldots,K\}$ and $l_j\neq l_{j+1}$ for $j=1,\ldots, r-1$.
Note that we call $t_1,\ldots,t_r$ the {\em change points of WSFs} since those timestamps represent when the jumps happen.
The mission of interest is estimating $s_1,\ldots,s_K$ from $x$. Note that in practice $x$ is inevitably contaminated by noise. In practice, the model is more complicated than that shown in \eqref{eq:WSF-special0}, and we will elaborate this part later.
Our contribution is summarized below. We propose a novel {\em iterative warping and clustering} algorithm that detects change points of WSFs and estimates $s_1,\ldots,s_K$ from the signal by combining three tools, including time-frequency analysis, singular value decomposition entropy and vector spectral clustering. By segmenting the signal under analysis into cycles and clustering the cycles, the change points $t_1,\ldots,t_r$ will be detected.
We demonstrate the efficiency of the proposed algorithm with both simulated and various real signals, including voice, ABP, ECG and accelerometer.
Moreover, we provide a theoretical justification of the proposed algorithm under the assumption that the amplitude and frequency of the signal are slowly time-varying and there are finite change points that model sudden changes from one WSF to the other, while the WSF could vary slowly between two change points. Specifically, we show that the iterative warping is guaranteed to convert the signal with time-varying amplitude and frequency into a signal whose fundamental component has amplitude $1$ and frequency $1$Hz, so that the WSF could be accurately truncated for the clustering, and hence the change point detection.

The paper is organized as follows. In Sec. \ref{section:model}, we review \texttt{WSFv1} and \texttt{WSFv2}, and introduce our new model that captures the sudden WSF changes. In Sec. \ref{Section:algo}, we introduce our proposed algorithm and explain its implementation details in Sec. \ref{sec:Practical}. In Sec. \ref{sec:Simulated}, we provide the numerical results, including simulated and real signals. The paper is closed with a discussion and conclusion in Sec. \ref{section:discussion}. The theoretical justification of the proposed algorithm is presented in the online Supplemental Material.

\section{Mathematical model}\label{section:model}
To handle an oscillatory signal with a non-sinusoidal oscillatory pattern, \texttt{WSFv1} was proposed in \cite{HTWu2013}, modeling the signal as
\begin{equation}\label{eq:WSF}
x(t) = A(t) s(\phi(t)),
\end{equation}
where $A(t)$ is a positive and smooth function called {\em instantaneous amplitude (IA)} (or amplitude modulation or time-varying amplitude), $\phi(t)$ is a strictly monotonically increasing and differentiable function called {\em phase function}, and $s(t)$ is a $1$-periodic function with mean 0 and unit $L^2$-norm called WSF that describes the oscillatory pattern of the signal. Usually we call the positive function $\phi'(t)$ {\em instantaneous frequency (IF)} (or time-varying frequency).
Note that the $1$-periodic function $s(t)$ with $\hat{s}(0)=0$ and $\|s\|_{L^2[0,1]}=1$ replaces the cosine function in the classical adaptive harmonic model \cite{Daubechies2011synchro}; that is, it allows the oscillation to be more complex than the simple sinusoidal function.
We need some \emph{slowly varying} conditions so that the \texttt{WSFv1} model  is identifiable \cite{chen2014non}. For a  small $\epsilon >0$, we assume
\begin{itemize}
\item[(O1)] $A\in C^1(\mathbb{R})$ and $|A'(t)| \leq \epsilon \phi'(t)$ for all $t \in \mathbb{R}$;
\item[(O2)] $\phi\in C^2(\mathbb{R})$, $|\phi''(t)|\leq \epsilon \phi'(t)$ for all $t \in \mathbb{R}$ and $\| \phi'' \|_{\infty} = M$ for some $M \geq 0$.
\end{itemize}
Note that in this model, {\em one} WSF is involved, which is stretched by $\phi(t)$ and scaled by $A(t)$; that is, there is only one non-sinusoidal oscillatory pattern.

While the \texttt{WSFv1} model has been applied to several problems \cite{hou2016extracting,Wu2016modeling}, it is limited and not suitable for many applications, particularly when the oscillatory pattern changes from time to time. See Fig. \ref{fig:ABP} for an ABP signal where the subject has arrhythmia so that we could observe an obvious change in the oscillatory pattern.
For a signal having this property, we say that it is generated by a {\em time-varying WSF}.
We need a more general model for such signal since it cannot be well modeled by \texttt{WSFv1} via scaling and stretching one WSF.

To address this challenge, in \cite{lin2016waveshape}, the authors proposed to generalize \texttt{WSFv1} in the following way to capture the time-varying WSF. Note that in \eqref{eq:WSF}, the WSF has the Fourier series expansion $s(t) = \sum_{l =1}^\infty a_l \cos(2\pi lt+b_l)$, where $a_l>0$ and $b_l\in(0,2\pi]$ are Fourier coefficients, and the convergence depends on the regularity of $s$. Note that since $\hat{s}(0)=0$ and $\|s\|_{L^2[0,1]}=1$, we have $\sum_{l=1}^\infty a_l^2=2$.
Thus,  the model \eqref{eq:WSF} becomes
\begin{equation}\label{eq:WSFrei}
x(t) =   \sum_{l =1}^\infty [a_lA(t)] \cos(2\pi l\phi(t)+b_l)\,.
\end{equation}
To capture the time-varying WSF, it is suggested in \cite{lin2016waveshape} to consider \texttt{WSFv2}, where the signal satisfies
\begin{equation}\label{eq:New_Model}
x(t) = \sum_{l = 1}^{\infty} B_l(t) \cos(2 \pi \phi_l(t))\,,
\end{equation}
where $A(t)a_l$ in \eqref{eq:WSFrei} is replaced by a non-negative function $B_l(t)$, $l\phi(t)+b_l/(2\pi)$ is replaced by a strictly monotonically increasing function $\phi_l(t)$, and $B_l(t)$ and $\phi_l(t)$ satisfy the following conditions:
\begin{itemize}
\item [{(T1-0)}] $B_l\in C^1(\mathbb{R})\cap L^\infty(\mathbb{R})$ for $l=1,2,\ldots$ and $B_1(t)>0$ and $B_l(t)\geq 0$ for all $t$ and $l=2,3\ldots$. $B_l(t) \leq c(l) B_{1}(t)$, for all  $t\in \mathbb{R}$ and $l = 1,2,\dots,\infty$, and with $ \{c(l)\}_{l = 1}^{\infty}$ a non-negative $\ell^{1}$ sequence. Moreover, there exists $N\in \mathbb{N}$ so that $\sum_{l=N+1}^\infty B_l(t)\leq \epsilon \sqrt{\sum_{l=1}^\infty B_l^2(t)}$ and $\sum_{l=N+1}^\infty lB_l(t)\leq D\sqrt{\sum_{l=1}^\infty B_l^2(t)}$ for some constant $D>0$.

\item [{(T2-0)}] $\phi_l\in C^2(\mathbb{R})$ and $|\phi'_l(t) - l\phi'_1(t)| \leq \epsilon \phi'_1(t)$, for all $t\in \mathbb{R}$ and  $l = 1,\dots,\infty$.

\item [{(T3-0)}]  $|B'_{l}(t)| \leq \epsilon c(l) \phi'_{1}(t)$ and $|\phi''_{l}(t)| \leq \epsilon l \phi'_{1}(t)$ for all $t \in \mathbb{R}$. Moreover, assume $\sup_{l;\,B_l\neq 0}\|\phi''_l\|_\infty=M$ for some $M\geq 0$.

\end{itemize}
In this model, we call $B_l(t) \cos(2 \pi \phi_l(t))$ the {\em $l$-th harmonic} of $x(t)$ for $l\geq 1$, and specifically call $B_1(t) \cos(2 \pi \phi_1(t))$ {\em fundamental component}. Note that
$ \{c(l)\}_{l = 1}^{\infty}$ in (T1-0) and (T3-0) is {\em related to} the Fourier coefficients of $s$ in \eqref{eq:WSFrei}. Clearly, \texttt{WSFv2}  is reduced to \texttt{WSFv1} when $B_l(t)=A(t)a_l$ with $\sum_{l=1}^\infty a_l^2=2$ and $\phi_l(t)=l\phi(t)+b_l/(2\pi)$, and in this case, $c(l)=a_l/a_1$ so that $\sum_{l=1}^\infty c(l)^2=2/a_1^2>2$. Also note that (T2-0) ensures that the IF of a harmonic is not far from an integer multiple of the fundamental frequency $\phi'_1(t)$.
In short, (T1-0)-(T3-0) allows the WSF to vary, but only {\em slowly} in the sense that the harmonics' behavior is well controlled by the fundamental component. We thus say that \texttt{WSFv2} is generated by {\em slowly varying WSFs}. See \cite{lin2016waveshape} for more discussion about this model.
This \texttt{WSFv2} model has been applied to model several signals for various purposes, like fetal electrocardiogram decomposition from the trans-abdominal electrocardiogram signal \cite{su2017extract}, lower body negative pressure analysis \cite{ShelleyAlianWu2022}, reconsideration of phase \cite{alian2022reconsider} and the accelerometer signal analysis \cite{wu2022application}.

We shall mention that the \texttt{WSFv2} introduced in \cite{lin2016waveshape} could be slightly generalized to include more general WSFs. Note that it is possible that a $1$-periodic function does not have a fundamental component \cite{Xu2018recursive,steinerberger2021fundamental}. In fact, note that $s(t) = \sum_{l =1}^\infty a_l \cos(2\pi lt+b_l)$ is $1$-periodic as long as $\textup{gcd}\{l|\,a_l\neq 0\}=1$.
We could thus generalize \texttt{WSFv2} introduced in \cite{lin2016waveshape} by replacing (T1-0)-(T3-0) by

\begin{itemize}

\item[(T1-1)] $B_l\in C^1(\mathbb{R})\cap L^\infty(\mathbb{R})$ for $l=1,2,\ldots$ and there exists $\ell\in \mathbb{N}$ so that $B_\ell(t)>0$ and $B_l(t)\geq 0$ for all $t\in \mathbb{R}$ and $l\neq \ell$.
For each time $t\in \mathbb{R}$, we have $\textup{gcd}\{l|\,B_l(t)> 0\}=1$ and $B_l(t) \leq c(l)B_{\ell}(t)$, for all  $l=1,2,\ldots$ and $t\in \mathbb{R}$, where $ \{c(l)\}_{l = 1}^{\infty}$ is a non-negative $\ell^{1}$ sequence. Moreover, there exists $N\in \mathbb{N}$ so that $\sum_{l=N+1}^\infty B_l(t)\leq \epsilon \sqrt{\sum_{l=1}^\infty B_l^2(t)}$ and $\sum_{l=N+1}^\infty lB_l(t)\leq D\sqrt{\sum_{l=1}^\infty B_l^2(t)}$ for some constant $D>0$.

\item[(T2-1)] $\phi_l\in C^2(\mathbb{R})$ and $|\phi'_l(t) - l\phi'_\ell(t)/\ell| \leq \epsilon \phi'_\ell(t)$ for all $t\in \mathbb{R}$ and  $l = 1,2,\dots$.

\item[(T3-1)]  $|B'_{l}(t)| \leq \epsilon c(l) \phi'_{\ell}(t)/\ell$ and $|\phi''_{l}(t)| \leq \epsilon l \phi'_\ell(t)/\ell$ for all $t \in \mathbb{R}$. Moreover, assume $\sup_{l;\,B_l\neq 0}\|\phi''_l\|_\infty=M$ for some $M\geq 0$.
\end{itemize}
Note that (T1-1) assumes that the $\ell$-th harmonic is dominant, and we call it the {\em dominant harmonic}.
When $\ell=1$, (T1-1)-(T3-1) is reduced to (T1-0)-(T3-0). In the following, we would still call this generalized model \texttt{WSFv2} and follow the same terminologies used for \eqref{eq:New_Model} to simplify the discussion.

However, \texttt{WSFv2} is still limited to quantify the signals shown in Figs. \ref{fig:ABP} and \ref{fig:ECG}, since the WSF variation in these signals is not ``slow''. We need a model that could capture the ``sudden change'' of the WSF. To this end, we introduce the following assumptions for \eqref{eq:New_Model} that further generalizes \texttt{WSFv2}. Specifically, we extend (T1-1)-(T3-1) and introduce

\begin{itemize}
\item[(T0)] There exist $r\in \mathbb{N}\cup\{0\}$ change points. If $r>0$, the change points are $t_i$, $i=1,\ldots,r$, so that $t_1<t_2<\ldots<t_r$. Set $I_i:=(t_i,\,t_{i+1}]$, where $i=1,\ldots,r-1$, $I_0=(-\infty,t_1]$ and $I_r=(t_r,\infty)$.

\item[(T1)] For each $i=0,\ldots,r$, $B_l\in C^1(I_i)\cap L^\infty(I_i)$ for $l=1,2,\ldots$  and there exists $\ell\in \mathbb{N}$ so that $B_\ell(t)>0$ and $B_l(t)\geq 0$ for all $l\neq \ell$.
For each time $t\in I_i$, we have $\textup{gcd}\{l|\,B_l(t)> 0\}=1$.
Also, there exists a non-negative $\ell^{1}$ sequence $\{c(l)\}_{l = 1}^{\infty}$ so that $B_l(t) \leq c(l)B_{\ell}(t)$ for all $l=1,2,\ldots$ and $t\in I_i$.
Moreover, there exists $N\in \mathbb{N}$ so that $\sum_{l=N+1}^\infty B_l(t)\leq \epsilon \sqrt{\sum_{l=1}^\infty B_l^2(t)}$ and $\sum_{l=N+1}^\infty lB_l(t)\leq D\sqrt{\sum_{l=1}^\infty B_l^2(t)}$ for some constant $D>0$ for all $t\in \mathbb{R}$.

\item[(T2)] $\phi_{\ell}(t)\in C^2(\mathbb{R})$. For each $i=0,\ldots,r$, $\phi_l\in C^2(I_i)$ and $|\phi'_l(t) - l\phi'_\ell(t)/\ell| \leq \epsilon \phi'_\ell(t)$ for all $t\in I_i$ and  $l = 1,2,\dots$.

\item[(T3)] $B_{\ell}(t)\in C^1(\mathbb{R})$. For each $i=0,\ldots,r$,  $|B'_{l}(t)| \leq \epsilon c(l) \phi'_{\ell}(t)/\ell$ and $|\phi''_{l}(t)| \leq \epsilon l \phi'_\ell(t)/\ell$ for all $t \in I_i$. Moreover, assume $\sup_{l;\,B_l\neq 0}\|\phi''_l\|_\infty=M$ for some $M\geq 0$.

\item[(T4)] If $r>0$, for each $i=1,\ldots,r$, there exist finite $\{l_i\}\subset \mathbb{N}$ so that $\ell\notin \{l_i\}$ and $B_{l_i}$ is discontinuous at $t_i$ so that $B_{l_i}-(\lim_{t\to t_i^+}B_{l_i}(t)-\lim_{t\to t_i^-}B_{l_i}(t))\chi_{[t_i,\infty)}$ is $C^1$ at $t_i$, and there exist finite $\{k_i\}\subset \mathbb{N}$ so that $\ell \notin \{k_i\}$ and $\phi_{k_i}$ is discontinuous at $t_i$  so that $\phi_{k_i}-(\lim_{t\to t_i^+}\phi_{k_i}(t)-\lim_{t\to t_i^-}\phi_{k_i}(t))\chi_{[t_i,\infty)}$ is $C^2$ at $t_i$.
\end{itemize}
We say a function satisfying \eqref{eq:New_Model} and (T0)-(T4) fulfills the \texttt{WSFv3} model.
Note that this seemingly complicated assumption is actually a simple extension of (T1-1)-(T1-3) to capture sudden WSF changes. Indeed, when $r=0$ in (T0), (T1)-(T3) are the same as (T1-1)-(T1-3) and \texttt{WSFv3} is reduced to \texttt{WSFv2}. When $r>0$, it says that the WSF could change ``fast'' at finite time points. The fast varying is described by the discontinuity of some non-dominant harmonic characterized by (T4), while we still require the dominant harmonic to be smooth. In others words, we assume that there are {\em at most} $r+1$ different types of WSF functions in the signal. The uniqueness of the decomposition is guaranteed by the identifiability of the model under mild conditions \cite{chen2014non}; that is, if  all WSFs have their Fourier series coefficients decaying fast enough, then the decomposition is unique up to an error of order $\epsilon$.

To further illustrate \texttt{WSFv3}, consider the following special case. Take $-\infty=t_0<t_1<t_2<\ldots<t_r<t_{r+1}=\infty$. Suppose $B_l(t)=A_{j}(t)a_{j,l}$ with $\sum_{l=1}^\infty a_{j,l}^2=2$ when $t\in (t_{j-1},t_j)$, $j=1,\ldots,r+1$, where $A_j$ satisfies (O1), and assume the fundamental component is the dominant harmonic. Also suppose $\phi_l(t)=l\phi(t)+b_{j,l}/(2\pi)$ when $t\in (t_{j-1},t_j)$, where $b_{1,1}=b_{1,2}=\ldots=b_{1,r+1}$. In this case, we can rewrite \eqref{eq:New_Model} as
\begin{equation}\label{eq:WSF-special}
x(t) = \sum_{j=1}^{r+1}A_j(t) s_j(\phi(t))\chi_{(t_{j-1},t_j]}(t)\,,\end{equation}
where $s_j(t)=\sum_{l=1}^\infty a_{j,l}\cos(2\pi lt+b_{j,l})$ is a $1$-periodic function and $\chi$ is the indicator function.
In other words, there are at most $r+1$ WSFs. Note that \eqref{eq:WSF-special0} is a simplification of \eqref{eq:WSF-special}.
A typical example of \texttt{WSFv3} that captures the {\em fast varying WSFs} is the ECG signal shown in Fig. \ref{fig:ECG}.
We shall finally mention that \texttt{WSFv3} could be viewed as a model in parallel to the wave-shape oscillatory model aiming to study the dynamics of WSFs \cite{YTLin2019wave}, which is designed to study signals from the time domain.

\section{Proposed algorithm}\label{Section:algo}
Based on \texttt{WSFv3}, we now introduce our algorithm, {\em iterative warping and clustering}.
The mission is to find all different WSFs from the signal $x$ satisfying  \texttt{WSFv3}. The basic idea of the algorithm is estimating the phase of the dominant harmonic, and use that phase to warp the signal for a better WSF estimation. To motivate the algorithm, suppose the fundamental component is dominant, and suppose we have obtained its phase ${\phi}_1$ via some estimation. We could consider a warping process to obtain
\begin{equation}\label{eq:New_Modelunwrapped}
\begin{aligned}
x(\phi^{-1}_1(t)) =& B_1(\phi^{-1}_1(t)) \cos(2 \pi t) \\
&+ \sum_{l = 2}^{\infty} B_l(\phi^{-1}_1(t)) \cos(2 \pi \phi_l(\phi^{-1}_1(t)))\,,
\end{aligned}
\end{equation}
As a special example, note that if \eqref{eq:WSF-special} is fulfilled, \eqref{eq:New_Modelunwrapped} is reduced to
{\footnotesize\begin{equation}\label{eq:WSF-special unwrapped}
x(\phi^{-1}_1(t)) = \sum_{j=1}^{r+1}A_j(\phi^{-1}_1(t)) s_j( t)\chi_{[\phi^{-1}_1(t_{j-1}),\,\phi^{-1}_1(t_j)]}(t)\,.
\end{equation}}
Thus, a simple rescaling by estimating the amplitude modulation $A_j(\phi^{-1}_1(t))$ allows us to recover $s_j$ by a demodulation. While this idea is simple, several details are needed in order to properly carry it out numerically. Below, we summarize the main steps in the continuous setup, and provide numerical details in the next section. The theoretical guarantees of the algorithm are presented in the online Supplemental Material.

\subsection{Estimation of amplitude and phase of dominant harmonic}\label{section:amfmest current}

The first step towards our goal is applying the short-time Fourier transform (STFT) to estimate the dominant harmonic and its amplitude and phase. We shall mention that while we could consider nonlinear approaches, like synchrosqueezing transform (SST) \cite{Daubechies1996,Daubechies2011synchro} and its higher-order variation \cite{Oberlin2015,Pham2017high}, we consider STFT to avoid distraction.
Let $x$ be a tempered distribution and $h(t)$ a real and even Schwartz function. The short-time Fourier transform of $x$ using a real-valued window function $h$ is defined as
\begin{equation}\label{eq:STFT}
V_x^{(h)}(t,\xi) = \langle x, h(\cdot-t) e^{i2 \pi \xi (\cdot-t)}\rangle\,,
\end{equation}
where $\langle \cdot,\cdot\rangle$ is the canonical duality pairing, and $t,\xi\in \mathbb{R}$ are time and frequency respectively.
When $x$ is a function, we have $V_x^{(h)}(t,\xi) = \int_{-\infty}^{+\infty} x(u) h(u-t) e^{-i2 \pi \xi (u-t)} du$.

Recall the following result.
Let $x(t)$ be a signal modeled as in \eqref{eq:New_Model} with (T1-1)-(T3-1) fulfilled and $h(t)$ be a real and even function such that $\hat{h}(\xi)$ is supported in $[-H,H]$, where $0<H<\inf_t \phi_1'(t)/2$, in the frequency domain. Define $Z_{l}(t) = [\phi'_{l}(t)-H,\phi'_{l}(t)+H]$. If $|\phi'_{l+1}-\phi'_l|>2H$, for $l=1,2,\ldots$, then for $\xi\in Z_{l}(t)$, we have
\begin{equation}\label{Traditional result of STFT}
V_x^{(h)}(t,\xi) \approx \frac{B_{l}(t)}{2} e^{i2\pi\phi_{l}(t)} \hat{h}(\xi- \phi'_{l}(t)).
\end{equation}
See \cite{delprat1992asymptotic} for a precise statement of this result and the quantification of $\approx$.
This result tells us that by picking a proper window $h(t)$, it is possible to obtain an estimation of $\phi'_\ell$, denoted as $\tilde{\phi}'_{\ell}(t)$,  via the ridge detection.
However, to our knowledge, existing results do not guarantee the regularity of the estimated IF, even under (T1)-(T3). In the online Supplemental Material, we provide a theoretical result that extends \eqref{Traditional result of STFT} to guarantee the regularity of the estimated IF $\tilde{\phi}'_{\ell}(t)$ under \texttt{WSFv3}.
Then, the AM could be estimated via
\begin{equation}\label{eq:Amplitude}
\tilde{B}_{\ell}(t) \approx \frac{2|V_x^{(h)}(t,\tilde{\phi}'_{\ell}(t))|}{|\hat{h}(0)|}.
\end{equation}
The phase function $\phi_\ell$ can be found by evaluating the phase of the complex function $V_x^{(h)}(t,\tilde{\phi}'_{\ell}(t))$ or  $\int_{|\xi-\tilde{\phi}'_{\ell}(t)|<H} V_x^{(h)}(t,\xi)d\xi$.

\begin{remark}

We shall remark that (T3) indicates that the IF varies slowly. However, not all signals fulfill such a strong assumption. Sometimes, the IF could vary fast, and such fact would lead to an imperfect estimate of phase. When the IF varies fast, we need different algorithms to estimate the phase and AM. For example, the second-order SST offers a more accurate IF estimation if the signal is a linear chirp, and when the phase departures from being quadratic, the higher-order SST can improve the accuracy.
We shall also mention that the ubiquitous presence of noise is another source of imperfect estimate of phase and AM. Luckily, it has been shown in the literature \cite{Brevdo2011,chen2014non,Yang2018,sourisseau2022inference} that the STFT and SST can robustly estimate IF from a noisy signal when the IF varies slowly as described in (T3).
To our knowledge, while a theoretical justification might be straightly established with the same techniques used in \cite{chen2014non,sourisseau2022inference}, it is still lacking, and numerically the high-order SST is still robust to noise. However, numerically higher orders seem to be more prone to noise-related instabilities. This however is out of the scope of this paper, and will be explored in our future work.
Note that it is natural to think of other widely applied decomposition algorithms, like continuous wavelet transform (CWT) \cite{Flandrin1998}, empirical mode decomposition (EMD) \cite{EMD1998}, empirical wavelet \cite{Gilles2013}, sliding singular spectrum analysis (SSA) \cite{Harmouche2017}, the Blaschke decomposition \cite{Nahon2000Thesis} and iterative filtering-based algorithms \cite{LinYang2009}, to estimate the amplitude and phase of the dominant harmonic. However, these algorithms are not suitable for our purpose since they are designed to decompose different oscillatory components, but not for segmenting signals nor to estimate different oscillatory patterns present in it. Specifically, EMD is limited to handle nonsinusoidal WSF, and hence dominant harmonic, since it is based on the local maxima; Blaschke decomposition is also limited to handle nonsinusoidal WSF since it is based on the winding behavior of all harmonics; CWT and empirical wavelet are limited in the high harmonics region due to the affine transform nature of wavelet; SSA might be limited when the IF is varying; iterative filtering-based algorithms need the spectral band of the component and might be limited when the IF is varying. It is possible to apply them in specific situations, but in general the result is less accurate.
\end{remark}

\subsection{Iterative warping and demodulation}

Suppose the $\ell$-th harmonic, where $\ell \geq 1$, is dominant. By assumption (T3), we know $\phi'_{\ell}(t)/\ell$ is close to $\phi'_1$. Thus, we could insert $\phi'_\ell(t)/\ell$  into \eqref{eq:WSF-special unwrapped} to finish the mission of warping.
In practice, however, $\phi'_\ell(t)/\ell$ is unknown and needs to be estimated, and the estimate $\tilde{\phi}'_{\ell}(t)$ is not perfect in general in the sense that it is not exactly equal to the true phase $\phi'_\ell(t)$, as is discussed in Sec. \ref{section:amfmest current}.

To resolve this imperfectness, in this paper we advocate the following iterative warping scheme. First, we obtain an estimation of the inverse of the phase $\psi_\ell^{[1]}(t)=\tilde{\phi}_\ell^{-1}(t) \approx \phi_\ell^{-1}(t)$ with the estimator via STFT shown in Sec. \ref{section:amfmest current}. Below, we take $\ell=1$ to simplify the discussion, but in practice we can take any $\ell\geq 1$ and replace  $\phi'_1(t)$ by  $\phi'_\ell(t)/\ell$ below and obtain the same conclusion. We proceed to set the \emph{first} warping and demodulation process and obtain
\begin{equation}\label{eq:first_dewarped}
{\small\begin{aligned}
\tilde{x}^{[1]}(t) =& \frac{x(\psi_{1}^{[1]}(t))}{B_{1}^{[1]}(\psi_{1}^{[1]}(t))} = \frac{B_1(\psi_{1}^{[1]}(t))}{B_{1}^{[1]}(\psi_{1}^{[1]}(t))} \cos\left(2 \pi t +  E_{1,1}(t)\right)\\
&+\sum_{n = 2}^\infty \frac{B_n(\psi_{1}^{[1]}(t))}{B_{1}^{[1]}(\psi_{1}^{[1]}(t))} \cos\left(n 2 \pi t +  E_{1,n}(t)\right),
\end{aligned}}
\end{equation}
where $E_{1,n}(t)$ represents the error caused by the imperfect estimation of the phase. We use the terminology \emph{warping} to refer to the composition of the signal with $\psi_{1}^{[1]}$ and \emph{demodulation} to refer to the division by the estimated amplitude $B_{1}^{[1]}(t)$.
We shall mention that a similar warping idea was considered in \cite{Xu2018recursive,Yang2019multiresolution}, but the purpose was decomposing the signal.
Due to the imperfect estimation of $\phi_1^{-1}$ and $B_1(\psi_{1}^{[1]}(t))$, clearly $E_{1,1}(t)\neq0$ and $\frac{B_1(\psi_{1}^{[1]}(t))}{B_{1}^{[1]}(\psi_{1}^{[1]}(t))}\neq 1$.
To handle this imperfectness, an intuitive approach is estimating the phase function again from $\tilde{x}^{[1]}(t)$, and hope that it could lead to an improved phase estimation.
However, in general the regularity of $E_{1,1}(t)$ and $\frac{B_1(\psi_{1}^{[1]}(t))}{B_{1}^{[1]}(\psi_{1}^{[1]}(t))}$ is not clear, which challenges the legality of applying the estimator in Sec. \ref{section:amfmest current}. Fortunately, as we will show in the theoretical section, this is not a problem.
Thus, from $\tilde{x}^{[1]}(t)$, the \emph{second} iteration of the phase estimation can be constructed and we obtain $\psi_1^{[2]}(t)$ as the estimator of the inverse of $2 \pi t +  E_{1,1}(t)$, and hence the second iteration signal $\tilde{x}^{[2]}(t) = x(\psi_1^{[1]} (\psi_1^{[2]}(t) ) )$. Then, the \emph{second} warped and demodulated signal would be
 \begin{equation}\label{eq:first_dewarped}
{\footnotesize\begin{aligned}
&\tilde{x}^{[2]}(t) = \frac{\tilde{x}^{[1]}(\psi_{1}^{[2]}(t))}{B_{1}^{[2]}(\psi_{1}^{[2]}(t))}\\
&=\frac{B_1(\psi_{1}^{[1]}(\psi_{1}^{[2]}(t)))}{B_{1}^{[1]}(\psi_{1}^{[1]}(\psi_{1}^{[2]}(t)))B_{1}^{[2]}(\psi_{1}^{[2]}(t))} \cos\left( 2 \pi t  + E_{2,1}(t)\right)\\
&+ \sum_{n = 2}^\infty \frac{B_n(\psi_{1}^{[1]}(\psi_{1}^{[2]}(t)))}{B_{1}^{[1]}(\psi_{1}^{[1]}(\psi_{1}^{[2]}(t)))B_{1}^{[2]}(\psi_{1}^{[2]}(t))} \cos\left(n 2 \pi t + E_{2,n}(t)\right)\,,
\end{aligned}}
\end{equation}
where $E_{2,n}(t)$ again represents the error. Ideally, we would expect that $|E_{2,1}(t)|< |E_{1,1}(t)|$ and $\left|\frac{B_1(\psi_{1}^{[1]}(\psi_{1}^{[2]}(t)))}{B_{1}^{[1]}(\psi_{1}^{[1]}(\psi_{1}^{[2]}(t)))B_{1}^{[2]}(\psi_{1}^{[2]}(t))}-1\right|< \left|\frac{B_1(\psi_{1}^{[1]}(t))}{B_{1}^{[1]}(\psi_{1}^{[1]}(t))} -1\right|$.
This procedure could be iterated, denoting the resulting signal after the $I$-th iteration as $\tilde{x}^{[I]}(t)$, and $\tilde{x}(t) = \lim_{I\to \infty}\tilde{x}^{[I]}(t)$ to be the final warped and demodulated signal. Ideally, we expect to obtain
\begin{equation}\label{eq:dewarped}
{\footnotesize
\begin{aligned}\tilde{x}(t) = \cos( 2 \pi t) + \sum_{n > 1} \frac{B_n(\phi_1^{-1}(t))}{B_{1}(\phi_1^{-1}(t))} \cos(n 2 \pi t + E_{\infty,n}(t))\,,
\end{aligned}}
\end{equation}
where $E_{\infty,n}(t)$ is the term coming from the time-varying WSFs; that is, $n 2 \pi t + E_{\infty,n}(t)=\phi_l(\phi^{-1}_1(t))$ shown in \eqref{eq:WSF-special unwrapped}.
In practice, we will only iterate a finite number of times, and the number of iterations will be discussed in Sec. \ref{sec:Practical} when we discuss the numerical implementation.

\subsection{Segmenting and Clustering}

Equation \eqref{eq:dewarped} says that the fundamental component of the warped and demodulated signal $\tilde{x}$ is $1$-periodic, so we can easily segment it into ``pieces'' of length equal to 1. Note that \eqref{eq:dewarped} is a special case of the model discussed in \cite{YTLin2019wave}, and we can write it as
\begin{equation}\label{eq:non_parametric}
\tilde{x}(t) = \sum_{k \in \mathbb{Z}} \delta(t-k) \ast s_k(t)=\sum_{k \in \mathbb{Z}} s_k(t-k)\,,
\end{equation}
where $\delta(\cdot)$ stands for the Dirac distribution, $\ast$ means convolution, and $s_k$ has $\mbox{supp}\{s_k\}\subseteq[-1/2,1/2]$ and belong to the class of analytic wave-shape functions.
Thus, we can construct
\begin{equation}\label{eq:M}
M_{\tilde{x}} = \{ \tilde{x}(t)|_{[k,k+1)};
\, k\in \mathbb{Z} \}\subset L^2([0,1])\,.
\end{equation}
That is, we collect signals over all segments $[k,k+1)$, and each signal represents, ideally, a WSF. To finish the mission of estimating all WSFs present in the signal, we find clusters in the set $M_{\tilde{x}}$, and each cluster represents one WSF. The segmentation and clustering of the cycles of the signal lead to an indirect detection of the points where the WSF changes. We shall mention that this segmentation step looks similar to the sliding SSA \cite{Harmouche2017}, but they are different. In sliding SSA, the signal is segmented into pieces via a sliding window, while we segment the signal into pieces according to its fundamental phase.
The numerical details will be addressed in Sec. \ref{sec:Practical}.

Last but not the least, we shall mention that an approach that uses the phase and the amplitude to estimate the waveforms of the signal can be found in \cite{Xu2018recursive,Yang2019multiresolution}. The algorithms proposed in \cite{Xu2018recursive,Yang2019multiresolution} assume the exact knowledge of $A$ and $\phi$, which is in general not trivial to obtain. Moreover, they do not propose a clustering of the set of the waveforms but estimate an \emph{average waveform} for the whole signal duration, and proceed recursively in a deflationary way.

\section{Numerical Implementation}\label{sec:Practical}

We now detail the numerical implementation of the proposed algorithm, which is summarized in Algorithm 1. Below we discuss the algorithm step by step.

\noindent \hrulefill

\textbf{Algorithm 1} 
\vspace{-3mm}

\noindent \hrulefill

\begin{algorithmic}[1]\label{algo1}
\STATE The input signal is $\mathbf x^{[0]}:=\mathbf x\in \mathbb{R}^N$.
\STATE Set $i=0$.
\WHILE {SVD entropy criterion not met}
\STATE Compute the STFT of $\mathbf x^{[i]}(t)$.
\STATE Detect the ridge by the method proposed in \cite{Colominas2020}.
\STATE Estimate the amplitude, denoted by $\tilde{B}^{[i]}_1(t)$, from Eq. \eqref{eq:Amplitude} and the frequency, denoted by $\tilde{\phi}_1^{[i]'}(t)$, as in \cite{Pham2019Eusipco}.
\STATE Estimate the phase via $\tilde{\phi}_1^{[i]}(t) = \int_0^t \tilde{\phi}_1^{[i]'}(u)du$.
\STATE Compute $\mathbf x^{[i+1]}$ from Eq.
\eqref{eq:dewarped}.

\STATE Construct the set $M_{\tilde{\mathbf x}}$ from $\tilde{\mathbf x}=\mathbf x^{[i+1]}$ as in Eq. \eqref{eq:M}.
\STATE Align all segments in $M_{\tilde{\mathbf x}}$ via synchronization, and obtain $\{\hat{\mathbf c}_l\}_{l=1}^P$.

\STATE $i\leftarrow i+1$.
\ENDWHILE

\STATE Determine the optimal number $K$ of clusters \cite{Calinski1974} for the set $\{\hat{\mathbf c}_l\}_{l=1}^P$ and cluster the data using the $k$-means algorithm \cite{kmeans}.
\STATE Estimate the waveforms as the cluster medians or by the trigonometric regression \cite{Brillinger1987,Kavalieris1994} on the cluster medians.
\end{algorithmic}
\vspace{-3mm}
\hrulefill

\subsection{Iterative warping}
Consider the signal $x(t)$ that is defined on $t\in[0,T]$. In practice, we work with the discretization of $x(t)$ with the sampling period $\Delta t = 1/f_s$, where $f_s>0$ is the sampling frequency. Thus, we have $\mathbf x\in \mathbb{R}^{N}$, where $N=\lfloor T f_s\rfloor$ and $\mathbf x[n] = x(n \Delta t)$. For the sake of simplicity, we drop the subindex for the dominant harmonic. Note that the warping process of estimating $x(\psi(t))$, with $\psi(t) = \phi^{-1}(t)$ being the inverse function of the phase, is equivalent to a \emph{non-uniform resampling} of $x[n]$.
To numerically carry out this warping process, we discretize the range of $\phi(t)$ over $[0,T]$ into $M$ uniform points with a separation of $\Delta \tau = (\phi(T)-\phi(0))/M$. Then, set $t_m$ so that $\phi(t_m) = m \Delta \tau$. By finding these points, we are actually estimating $\phi^{-1}(t)$, so the warped signal is $\mathbf x^{[1]}\in \mathbb{R}^{M}$ so that $\mathbf x^{[1]}[m] = x(t_m)$. We set $M=N$, but it can be other numbers if needed.
See Fig. \ref{fig:Warping} for an illustration of this process.
The same warping is carried out on $\mathbf x^{[1]}$, and hence iteratively for $I$ steps. Denote $\tilde{\mathbf x}:=\mathbf x^{[I]}$. How to determine $I$ and the benefits of this iterative procedure will be illustrated in the following steps.

\begin{figure}[t]
\begin{center}
\includegraphics[width=\columnwidth]{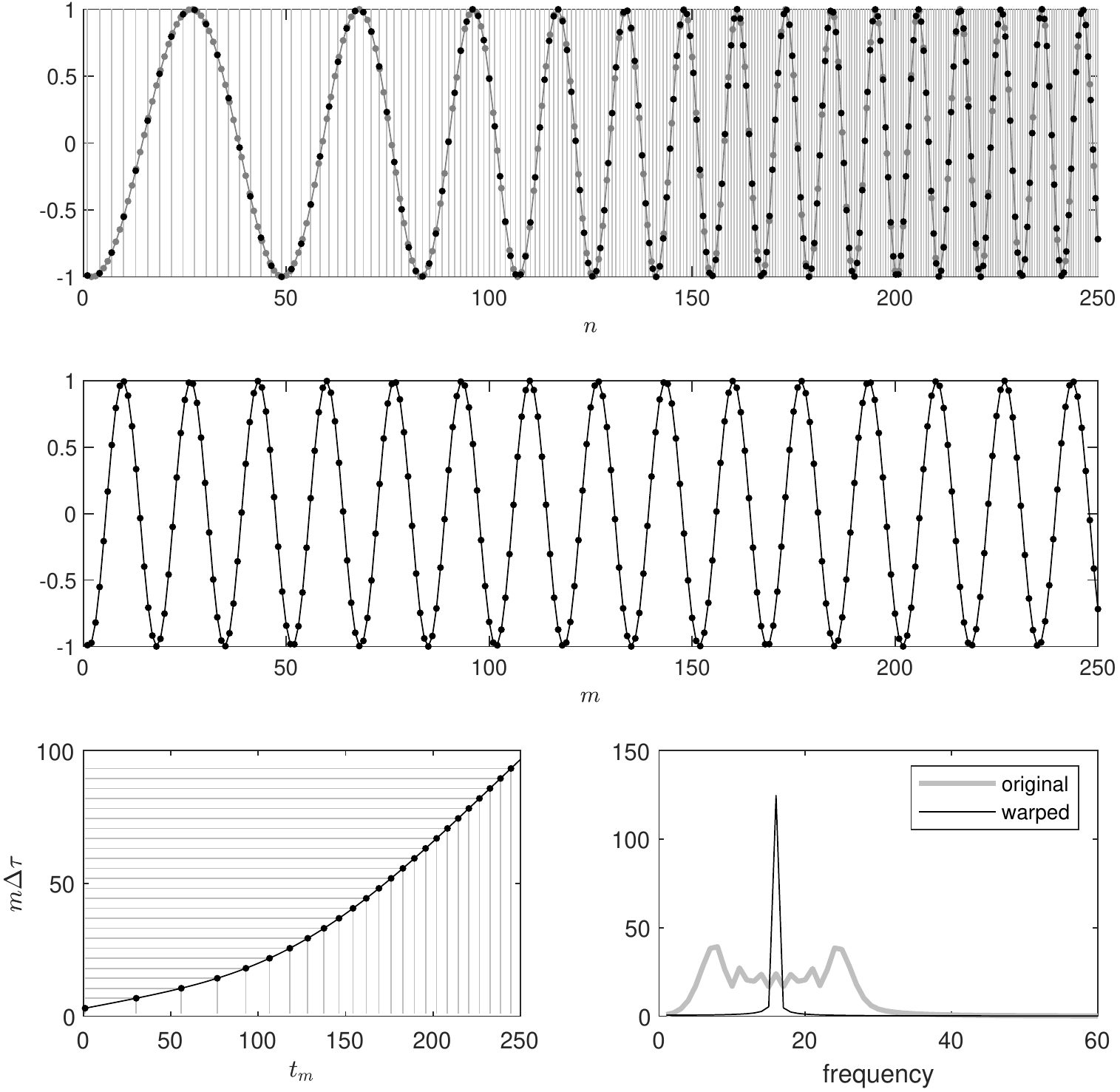}
\end{center}
\caption{An illustration of the warping process. Top: the original signal with its samples marked by gray dots, where the superimposed gray vertical lines indicate $t_m$, and the resampled points of the original signal shown in black dots. Middle: the warped signal. Bottom left: $\phi(t)$ and the associated $t_m$. Bottom right: the spectra of the original and warped signals.}
\label{fig:Warping}
\end{figure}

\subsection{Segmenting}
We segment the warped signal $\tilde{\mathbf x}^{[i]}$ after $i$ iterations in the following way. Ideally, after the iterative warping, the period of the fundamental component of the signal $\tilde{\mathbf x}^{[i]}$ is $1$. Suppose there are $P\in \mathbb{N}$ cycles of period $1$, each with length $L$. We divide $\tilde{\mathbf x}^{[i]}$ into $P$ segments,  $c_p\in \mathbb{R}^L$, $p = 1,2,\dots,P$,  and allocate them in a data matrix $\mathbf{M} \in \mathbb{R}^{P \times L}$:
\begin{equation}
\mathbf{M} =
\begin{bmatrix}
\rule[3pt]{20pt}{.5pt}&\mathbf c_1&\rule[3pt]{20pt}{.5pt} \\
\rule[3pt]{20pt}{.5pt}&\mathbf c_2&\rule[3pt]{20pt}{.5pt} \\
&\vdots& \\
\rule[3pt]{20pt}{.5pt}&\mathbf c_P&\rule[3pt]{20pt}{.5pt}
\end{bmatrix}.
\end{equation}

\subsection{SVD Entropy}
To answer the question of how many iterations are needed, we consider the singular value decomposition (SVD) entropy \cite{Alter2000} that has been used in a time-frequency context \cite{Boashash2013}.
The basic idea of SVD entropy is measuring how different its rows (or columns) are from each other. It is defined as the Shannon entropy of the normalized singular values of the matrix, which can be easily found with a singular value decomposition \cite{Bishop2006}. Specifically, if the singular values of $\mathbf{M}$ are $\sigma_1,\ldots,\sigma_P$, the SVD entropy is defined as
\begin{equation}
S_{\mathbf{M}} = -\sum_{i=1}^P \frac{\sigma_i}{\sum_{u=1}^P \sigma_u} \log\left(\frac{\sigma_i}{\sum_{u=1}^P \sigma_u}\right).
\end{equation}
Clearly, if all the rows are multiples of each other; that is, $\mathbf{M}$ is a rank-1 matrix, then only one singular value would be different from zero, and the SVD entropy would attain its minimum value, where we follow the convention and set $0\times \infty=0$. On the other hand, if all the rows are completely different from each other and all the singular values have the same weight, then the SVD entropy would be large. In our case, there are a few WSFs in the signal. By the slowly varying assumption of WSF, if all cycles are aligned properly, the rows are clustered into subsets so that all rows in one cluster are similar and only a few singular values are dominant, which leads to a low SVD entropy. If they are not aligned, the temporal shift would lead to a spreading of singular values, and hence to a high SVD entropy. In a special case that the period of $\tilde{\mathbf x}^{[i]}$ is precisely $1$, we have only one WSF, and all cycles are aligned properly, all rows are similar, then we have lowest SVD entropy. In other words, a low SVD entropy indicates that we have ran sufficient iterations. We thus suggest stopping the iteration process when the SVD entropy is stagnated.

\subsection{Synchronization}
In practice, the iterative warping might not be enough, particularly when there are change points of WSFs. Such change points might deviate the global phase, which leads to segments not properly aligned. To handle this case, we propose to apply graph connection Laplacian (GCL) based synchronization \cite{bandeira2013cheeger,marchesini2016alternating}. First, solve the optimization problem
\begin{equation}
R_{ij} = \arg\min_{R \in Z_L} \| R \circ \mathbf c_i - \mathbf c_j \|_{\mathbb{R}^L},
\end{equation}
where $Z_L\simeq \mathbb{Z}/L\mathbb{Z}=\{e=g^0, g, g^2, \ldots, g^{L-1}\}$ is the cyclic group and $(R\circ \mathbf c_i)[k]=\mathbf c_i[(k+l)\mod L]$ when $R=g^l$.
Note that we can represent $g^l\in Z_L$ in the matrix form
$g^l = \begin{bmatrix}
\cos(2\pi l/L) \, \, \sin(2\pi l/L) \\
-\sin(2\pi l/L) \, \, \cos(2\pi l/L)
\end{bmatrix}\in SO(2)$.
In case the minimum is not unique, we pick the one with the minimal angle. Thus, it is clear that $R_{ij} = R_{ji}^T \in SO(2)$. Since there are $P$ segments, we construct a $P \times P$ block matrix $\mathcal{C} \in \mathbb{R}^{2P\times 2P}$, with $2 \times 2$ blocks, where $\mathcal{C}[i,j] = R_{ij}$.
Denote the eigendecomposition $\mathcal{C} = U \Lambda U^T$, where the eigenvalues are sorted in the decreasing order with $U=\begin{bmatrix} u_1 & u_2 &\ldots & u_{2P}\end{bmatrix}\in O(2P)$. Then
\begin{equation}
Q = \begin{bmatrix}
u_1 \, \, u_2
\end{bmatrix}
\in \mathbb{R}^{2P \times 2}
=
\begin{bmatrix}
Q_1^T &
\ldots &
Q_P
\end{bmatrix}^T,
\end{equation}
with $Q_l \in \mathbb{R}^{2 \times 2}$. A rotation matrix $\hat{R}_i$ is fit to $Q_i$ by performing an SVD $Q_i = \Phi_i S_i \Psi_i^T$ and setting $\hat{R}_i = \Phi_i \Psi_i^T$. %
Finally, by setting $\hat{R}_1$ as the baseline, the other segments are aligned by applying
\begin{equation}
\hat{\mathbf c}_l := \hat{R}_1 \circ \hat{R}_l^T \circ \mathbf c_l,
\end{equation}
which leads to a synchronized version of the matrix. We shall mention that theoretically, if $\mathbf c_j$ comes from some shift $R_j\in Z_M$ of a template $\mathbf c$, then it is $\mathbf c_j=R_j\circ \mathbf c$. In this ideal case, it is clear that $R_{ij}=R_iR_j^T$, and hence $\mathcal{C}$ is a rank 2 matrix, and $Q=[R_1^T,\ R_2^T,\ldots,R_P^T]^T$, which allows us to fully recover the shifts.

\subsection{Clustering and change point detection}
Finally, we run a clustering algorithm on $\{\hat{\mathbf c}_l\}_{l=1}^P$ to obtain all patterns of WSFs.
There are several options to cluster a given data set. For instance, $k$-means \cite{kmeans}, agglomerative hierarchical cluster trees \cite{Ward1963,Rokach2005} or Gaussian mixture models \cite{Bishop2006}. In this work, we consider the $k$-means \cite{kmeans} since it performs consistently better. We suggest running the $k$-means with 50 replicates and conclude the result with the lowest sum of within-cluster distances. The optimal number of clusters $k$ was chosen using the Calinski-Harabasz criterion \cite{Calinski1974}.
Finally, we estimate the waveforms as the cluster medians or by the trigonometric regression \cite{Brillinger1987,Kavalieris1994} on the clusters medians.
The trigonometric regression on the cluster median is obtained by evaluating $\{\hat{a}_1,\ldots,\hat{a}_K,\hat{b}_1,\ldots,\hat{b}_K\} = \arg \min_{a_k,b_k} \|m(t) - \sum_{k=1}^K (a_k \cos(2 \pi k t) + b_k \sin(2\pi k t) ) \|_{L^2[0,1]}$, where $m(t)$ is the median of the given cluster and $K$ is determined by the criterion from \cite{ruiz2022wave}. and setting $m_R(t):=\sum_{k=1}^K (\hat a_k \cos(2 \pi k t) + \hat b_k \sin(2\pi k t) )$ as the solution.
We mention that empirically this trigonometric regression approach consistently gives a better result compared with the naive median when the noise level is high. Last but the not least, once the clusters are determined, the change points of WSFs are determined by the jumps from one cluster to another.

\section{Numerical results}\label{sec:Simulated}

We show various examples, ranging from the simplest simulated case, to the complicated real world signals. The Matlab code to reproduce all results in this section could be found in https://github.com/macolominas/WarpingWSF.

\subsection{Simulated Signal}
We start from the following simulated signal:
\begin{equation}\label{eq:Artificial}
{\footnotesize
\begin{aligned}
x_\phi(t) = \cos(2\pi \phi(t)) + A(t) \cos(4\pi \phi(t)) + B(t) \cos(6\pi\phi(t)),
\end{aligned}}
\end{equation}
where $t \in [0,1]$, $A(t) = \left( \frac{1}{1+e^{-250(t-1/3)}} - \frac{1}{1+e^{-250(t-2/3)}}  \right)$, $B(t) = \left( \frac{1}{1+e^{-250(t-1/3)}} \right)$, and
$\phi(t) = 40 t + \frac{1}{2\pi} \cos(8\pi t)$.
The sampling rate is 6000 Hz. The signal is then contaminated by the zero-mean Gaussian white noise with a signal to noise ratio (SNR) 0dB.
The results are shown in Fig. \ref{fig:Artificial_A}. 

\begin{figure}[t]
\begin{center}
\includegraphics[width=\columnwidth]{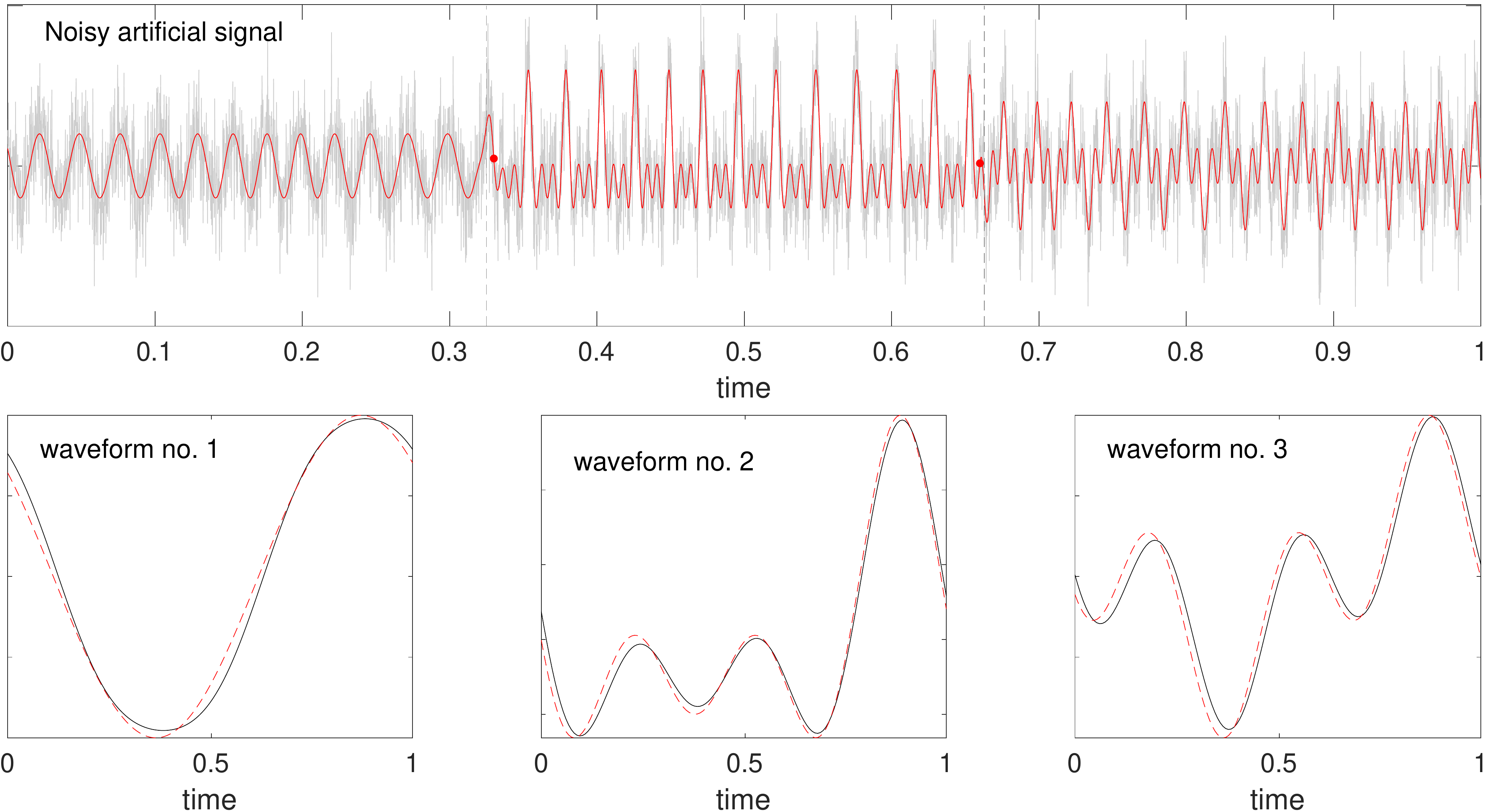}
\end{center}
\caption{Top: the simulated noisy signal of Eq. \eqref{eq:Artificial} is plotted as the gray curve, with the detected change points as vertical dashed lines, and the ground truth clean signal is superimposed as the red curve with the true change points as red dots. Bottom: the estimated wave shape functions (black) along with the ground truths (red).}
\label{fig:Artificial_A}
\end{figure}

Next, we illustrate the benefits of the iterative warping. We synthesize copies of
\begin{equation}
y_{ij}(t) = x_{\phi_i}(t) + \sigma n_j(t)\,,
\end{equation}
where $x_{\phi_i}$ is defined in \eqref{eq:Artificial}, $n_j(t)$ is the Gaussian white noise with zero mean and unit variance, $\sigma>0$ is determined by the desired SNR, and the phase $\phi_i$ is also randomly generated in the following way that is independent of $n_j$. Take
$\phi_i(t) = \phi(t) + Y_i(t)$,
where $t\in [0,1]$, $Y_i$ is the $i$-th independent realization of the smoothed Brownian motion $Y$ defined as
$Y(t) := \int_{0}^t \frac{X(u)}{\|X\|_{L^\infty[0,1]}}  du$,
$X(t) = W\ast K_B(t)$, $W$ is the standard Brownian motion, $K_B(t)$ is the Gaussian function with the standard deviation $B>0$ and $\ast$ denotes the convolution operator.
We consider three different SNRs, 30 dB, 20 dB, and 10 dB, and realize $100$ $y_{ij}(t)$ for each SNR. We perform up to three iterative warpings, and measure the quality of the achieved data matrix with the SVD entropy. Estimation error of the actual waveforms by the cluster median is computed as the root mean square error (RMSE). The results are shown in Fig. \ref{fig:Iterations}. It can be clearly observed how the estimation errors decrease with the iterations, as the SVD entropy does. This illustrates the benefits of our iterative scheme, which leads to a better waveform estimation.

\begin{figure}[t]
\begin{center}
\includegraphics[width=\columnwidth]{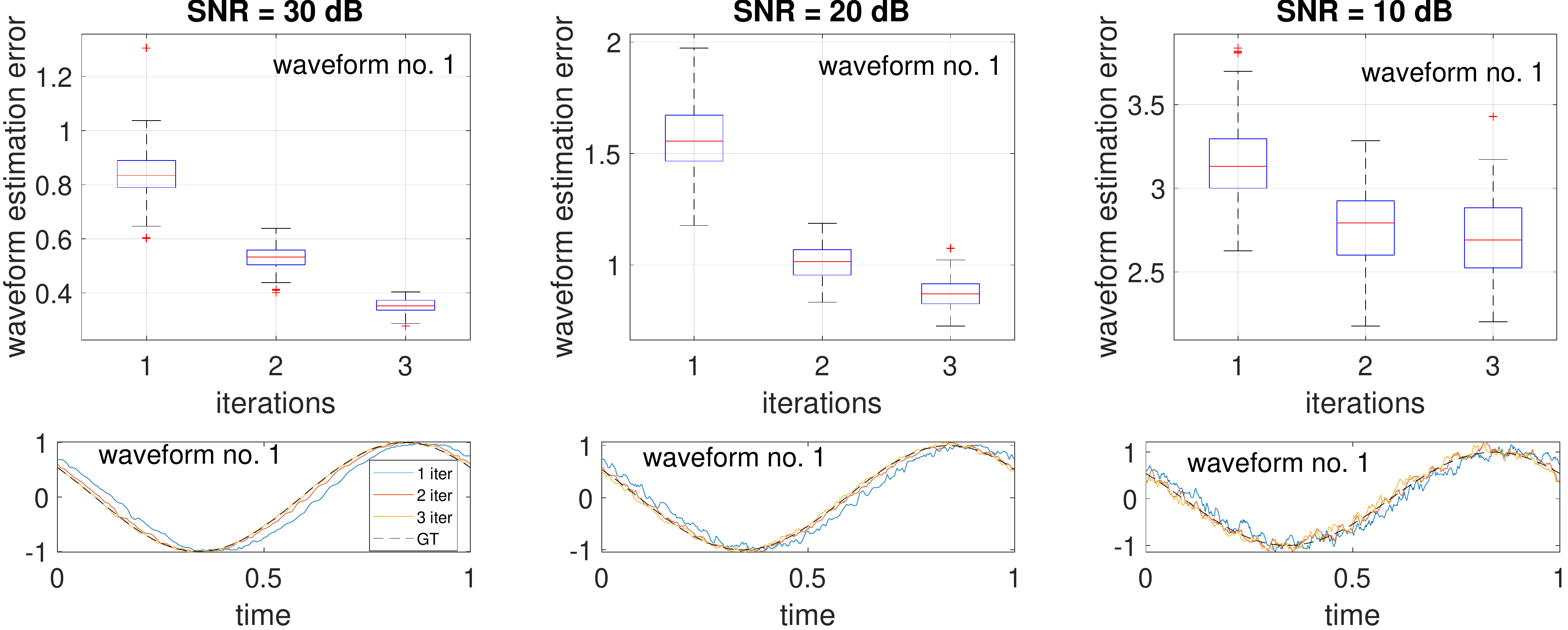}
\vspace{2mm}
\includegraphics[width=\columnwidth]{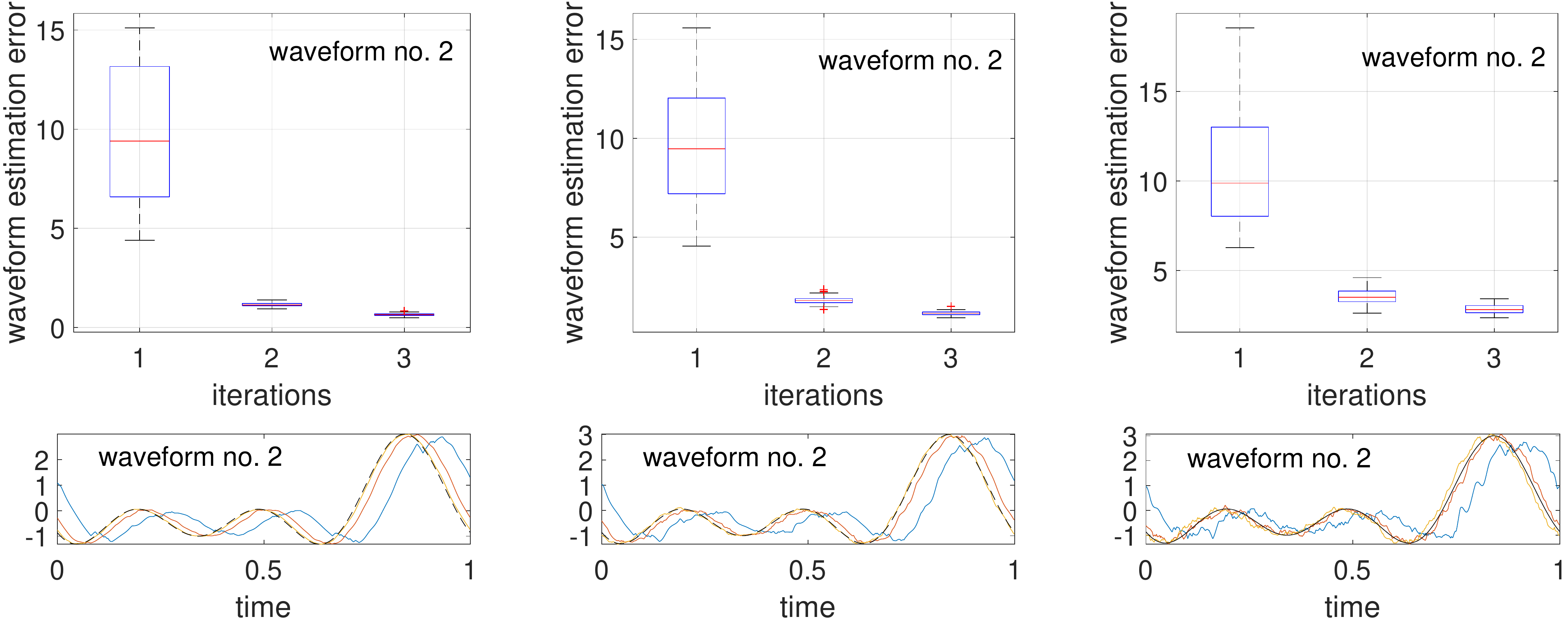}
\vspace{2mm}
\includegraphics[width=\columnwidth]{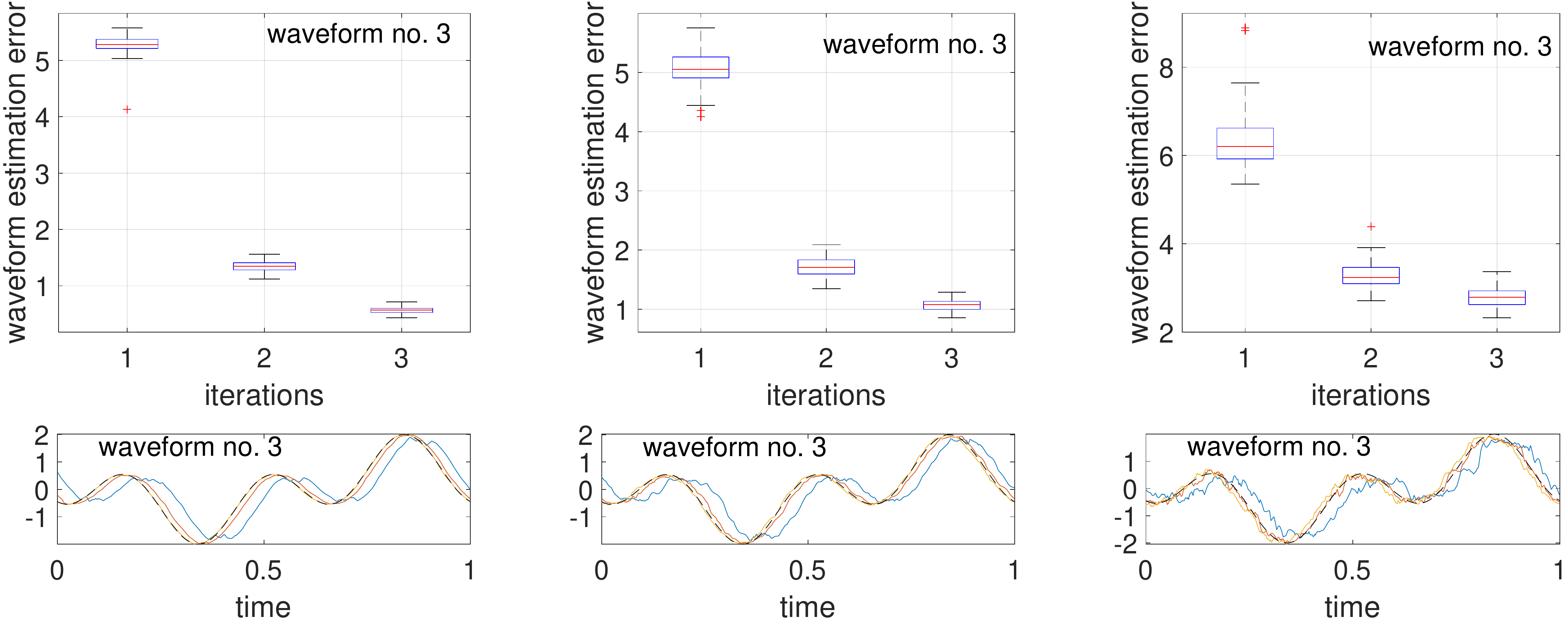}
\vspace{2mm}
\includegraphics[width=\columnwidth]{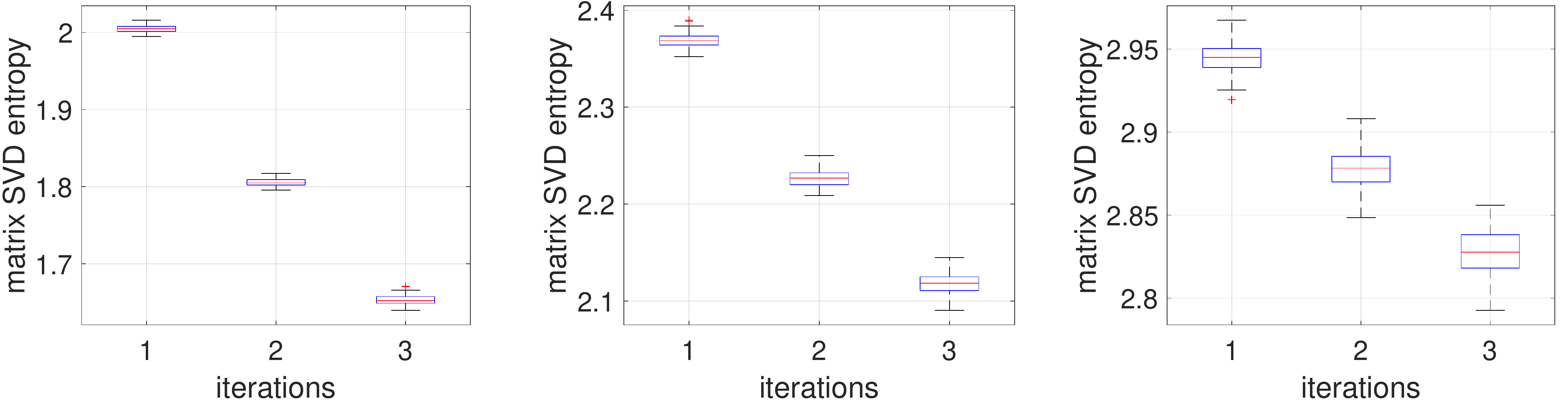}
\end{center}
\caption{WSF estimation error and matrix SVD entropy for the first, second, and third iteration. Left to right: the noise levels are 30 dB, 20 dB, and 10 dB respectively. Top to bottom: the estimation errors, the estimated WSFs and the SVD entropies of the first, second and third WSFs.}
\label{fig:Iterations}
\end{figure}

In order to study the performance of our algorithm for the change point detection task, we consider a detected change point as a true positive if it is within $\pm1$ cycle of the ground truth. Similarly, we define false positive and false negative. The F1 score is reported in Table 1. Upon the true positive, we computed the RMSE of the difference between the ground truth and the estimated change point, and report the results also in Table 1. The result suggests that our algorithm is robust to various noise levels, and the F1 score clearly improves with the iterations.
\begin{table}[]
\begin{center}
\begin{tabular}{ll|l|l|l|}
\cline{3-5}
                                                                       &         & 1 iter & 2 iter & 3 iter \\ \hline
\multicolumn{1}{|l|}{\multirow{3}{*}{\rotatebox[origin=c]{90}{30 dB}}} & F1 score& 0.5    & 1      & 1 \\ \cline{2-5}
\multicolumn{1}{|l|}{}                                                 & RMSE 1  & 0.002  & 0.0007 & 0.0002 \\ \cline{2-5}
\multicolumn{1}{|l|}{}                                                 & RMSE 2  & 0.172  & 0.005  & 0.0044 \\ \hline \hline
\multicolumn{1}{|l|}{\multirow{3}{*}{\rotatebox[origin=c]{90}{20 dB}}} & F1 score& 0.5    & 1      & 1 \\ \cline{2-5}
\multicolumn{1}{|l|}{}                                                 & RMSE 1  & 0.002  & 0.0007 & 0.0002 \\ \cline{2-5}
\multicolumn{1}{|l|}{}                                                 & RMSE 2  & 0.0172 & 0.005  & 0.0044 \\ \hline \hline
\multicolumn{1}{|l|}{\multirow{3}{*}{\rotatebox[origin=c]{90}{10 dB}}} & F1 score& 0.4981 & 1      & 1 \\ \cline{2-5}
\multicolumn{1}{|l|}{}                                                 & RMSE 1  & 0.002  & 0.0007 & 0.0003 \\ \cline{2-5}
\multicolumn{1}{|l|}{}                                                 & RMSE 2  & 0.0173 & 0.0049 & 0.0044 \\ \hline
\end{tabular}
\end{center}
\caption{F1 score and RMSEs for change point detection task.}
\end{table}

\subsection{Voice (/a/ vowel and concatenated vowels)}
The real-world signal is a recording from the Saarbruecken Voice Database \cite{Saarbruecken}. It corresponds to a sustained /a/ vowel from a 22-year-old healthy female speaker (recording session 62) of the type ``low-high-low'', meaning there is a change on its pitch. The spectrogram of this signal is shown on the top left panel of Fig. \ref{fig:Voz}, where the detected ridge is superimposed as a red curve. Besides the change of the pitch, we can observe an oscillation of the IF on the final third of the signal.
The warped and demodulated signal, along with its corresponding spectrogram, is shown on the second row of Fig. \ref{fig:Voz}. The WSFs extracted from the warped and demodulated signal (one cycle per matrix row) are shown on the second row of Fig. \ref{fig:Voz}, along with the results of the clustering. Two clusters are found for this signal, and the estimated WSFs are shown at the right column of Fig. \ref{fig:Voz}. We also offer a zoomed version of the signal, on the transition zone around time $t = 0.75$ s, where we can see that the WSF change is associated with an increase in the frequency. For a comparison purpose, we show a zoomed version of the warped signal, which makes clear that the WSF has changed, but the warped frequency remains almost constant.

This example suggests an interesting fact about the dynamics of the sustained vowels when there is a change in the pitch. Intuitively, one might think that a change in the pitch would only produce a compression or dilation of the WSF since it is the same vowel that is uttered. However, we also see a change in the morphology of WSF. This is not surprising since the pitch change depends on a muscle contraction of the phonating system, which consequently changes the WSF.

\begin{figure}[t]
\begin{center}
\includegraphics[width=\columnwidth]{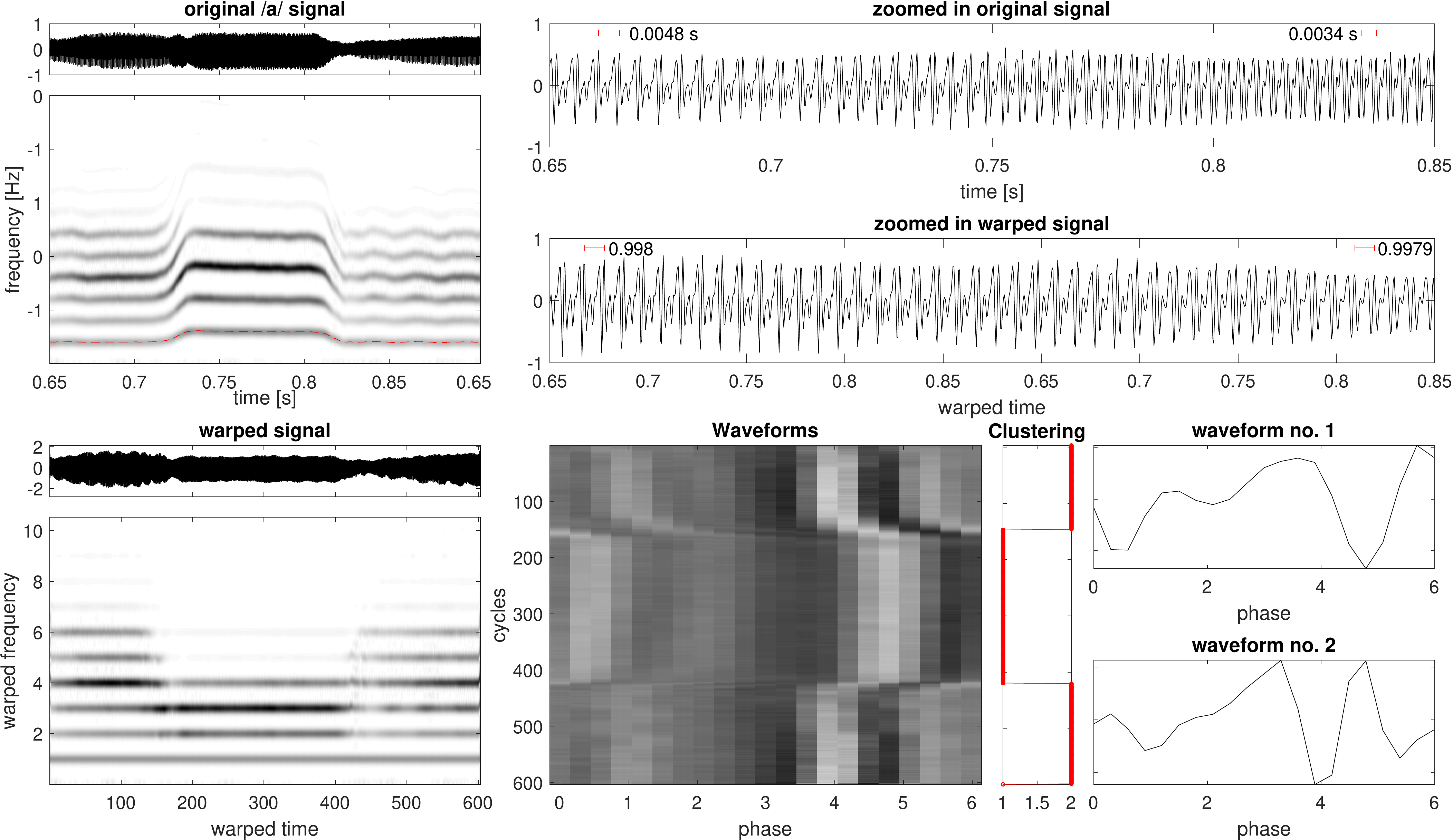}
\end{center}
\caption{Voice signal from \cite{Saarbruecken}. Top left: the original signal and its spectrogram with the detected ridge superimposed as the red curve. Top right: the zoomed in signal for a better visualization. Bottom left: the warped and demodulated signal and its spectrogram. Bottom middle: the data matrix of the extracted cycles from the warped and demodulated voice signal (the darker point indicates the larger value). Bottom right: results of the proposed algorithm.}
\label{fig:Voz}
\end{figure}

\begin{figure}[t]
\begin{center}
\includegraphics[width=\columnwidth]{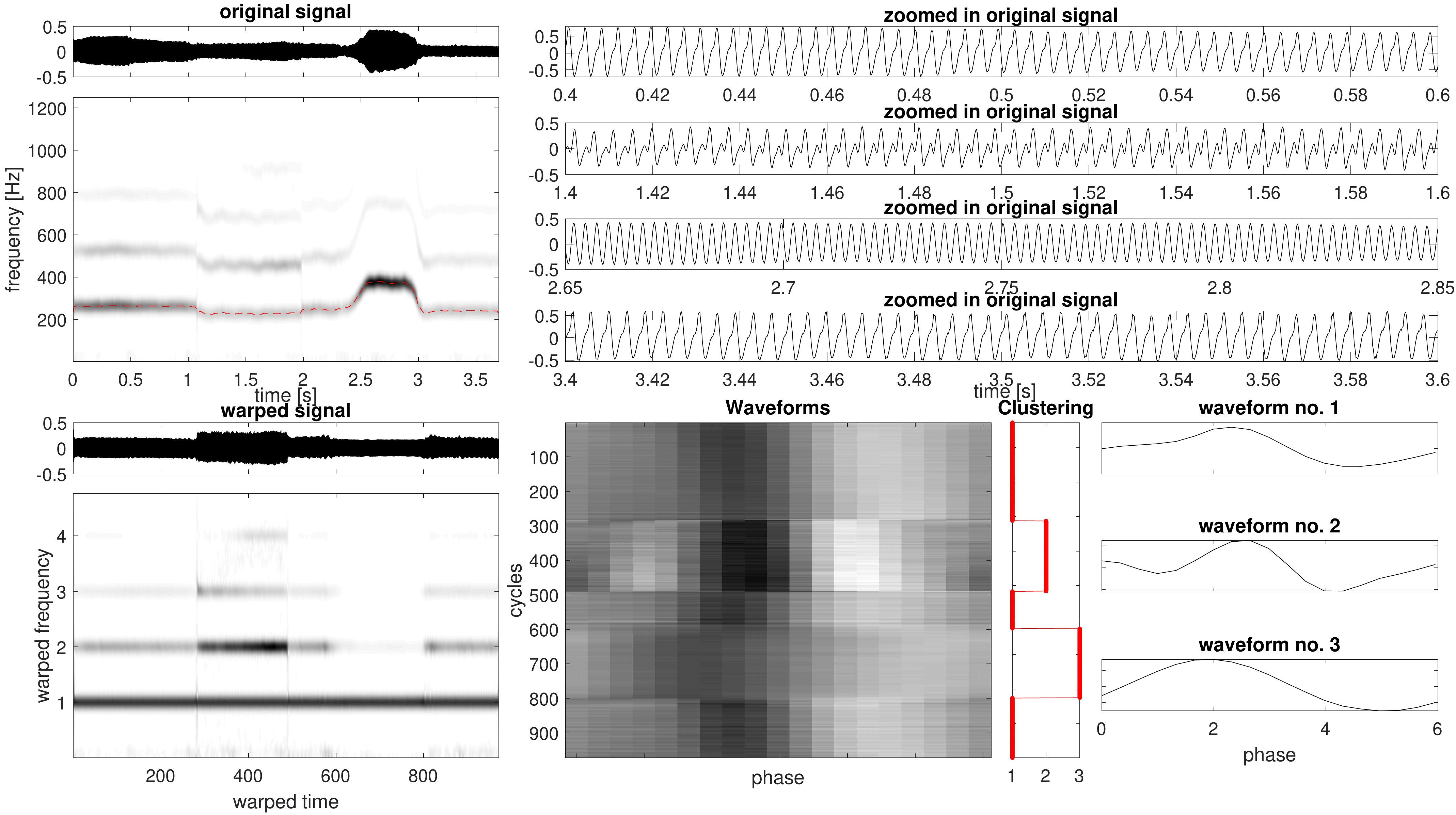}
\end{center}
\caption{Concatenated voice signal from \cite{Saarbruecken}. Top left: the original signal and its spectrogram with the detected ridge superimposed as the red curve. Top right: the zoomed in signal for a better visualization. Bottom left: the warped and demodulated signal and its spectrogram. Bottom middle: the data matrix of the extracted cycles from the warped and demodulated voice signal (the darker point indicates the larger value). Bottom right: results of the proposed algorithm.}
\label{fig:Voz_CONCATENATED}
\end{figure}

Next, we create a more challenging example by concatenating different signals from the same database. Taking the recordings of a 19-year-old healthy female speaker (recording session 8), we concatenate a ``neutral'' /i/ vowel, a ``low'' /u/ vowel, and a ``low-high-low'' /i/ vowel. We test the capabilities of our proposal to detect different vowels within a given signal. The results are shown in Fig. \ref{fig:Voz_CONCATENATED}. We see that three different waveforms are detected, including the low or neutral /i/ (no differences between them), the low /u/, and the high /i/. When the speaker is asked to utter the low-high-low vowel, the low level is the same as the neutral level. The major difference can be found when uttering the high vowel, which again comes from the configuration change of the phonating system, and hence the WSF.

\subsection{Arterial Blood Pressure}
We consider here an ABP signal from the `a70' recording in the 2010 CinC/PhysioNet Challenge database \cite{Moody2010physionet,Physionet}. We also show a simultaneous ECG recording, which makes clear the presence of an arrhythmia.
In this example, we perform 5 iterations, which is determined by the SVD entropy.
See Fig. \ref{fig:ABP} for the results. The warped and demodulated signal, and the spectrograms are shown on the first two rows (the detected ridge is shown in red color).
We see that there is a dramatic variation in the IF
(second row, left column) around 60 to 90 s. This deviation might cause the phase shift of the WSF before and after. This fact can be seen in the data matrix shown on the second row, middle column.

On the third row, we present the results with three iterations, where the phase shifting problem is clearly solved. Specifically, those cycles around warped time 90 s to 100 s and those cycles before warped time 60 s belong to the same cluster.
The clustering results reveal 5 clusters, detecting a WSF that presents a strong ``second'' peak within a cycle, which represents arrhythmia with premature atrial contraction. This arrhythmia is clearly observed on the zoomed in ECG shown at the top of the figure.

On the fourth row, we present the result after five iterations. The decreasing SVD entropy reveals that the obtained clustering is better than before, although the estimated WSFs do not significantly differ from the previous case, and the arrhythmia is captured. On the fifth row, we present the synchronized data matrix. Here, the clustering performance is better, and the four clusters are representatives of the actual WSFs present in the signal.

\begin{figure}[t]
\begin{center}
\includegraphics[width=\columnwidth]{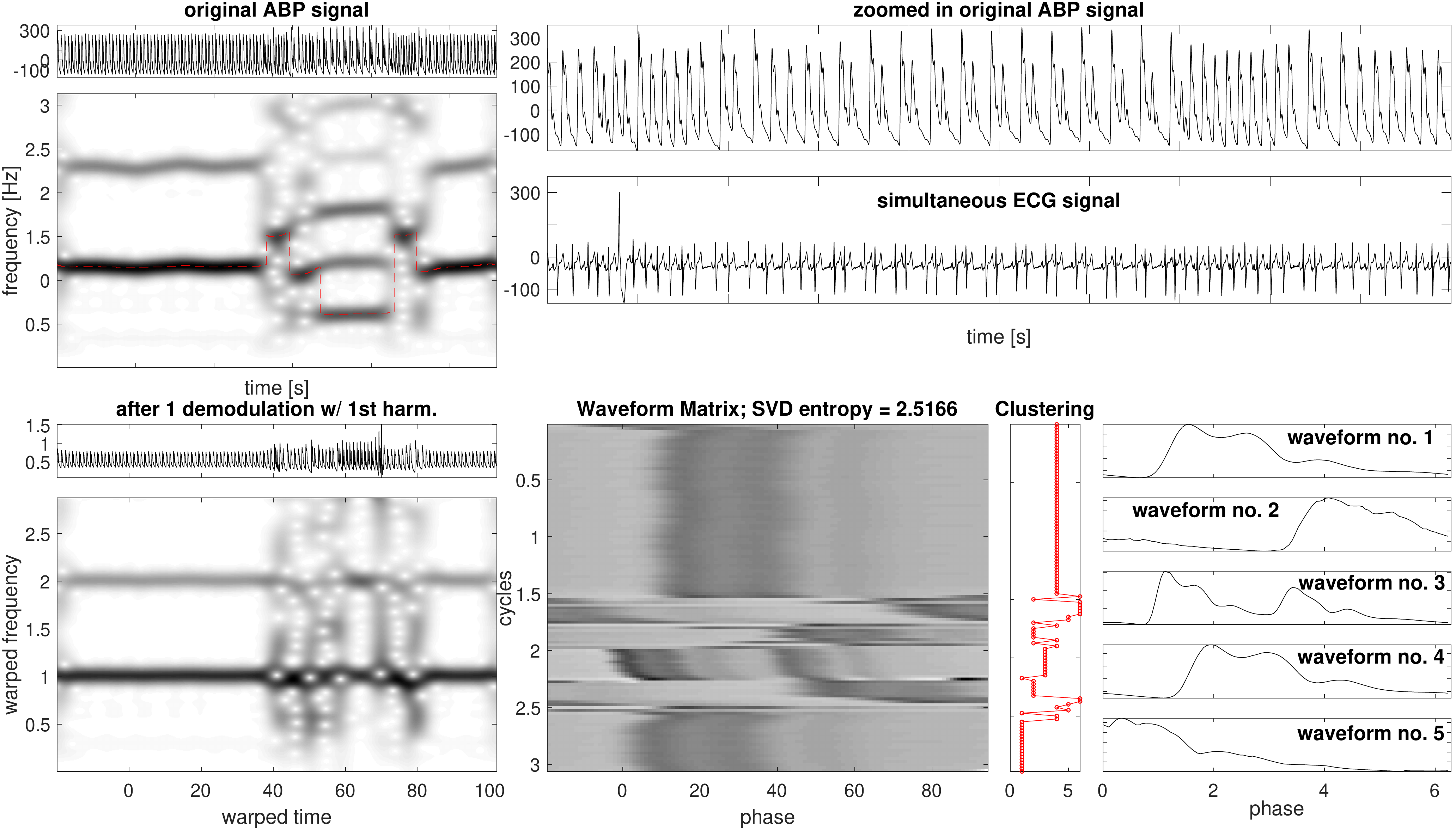}
\includegraphics[width=\columnwidth]{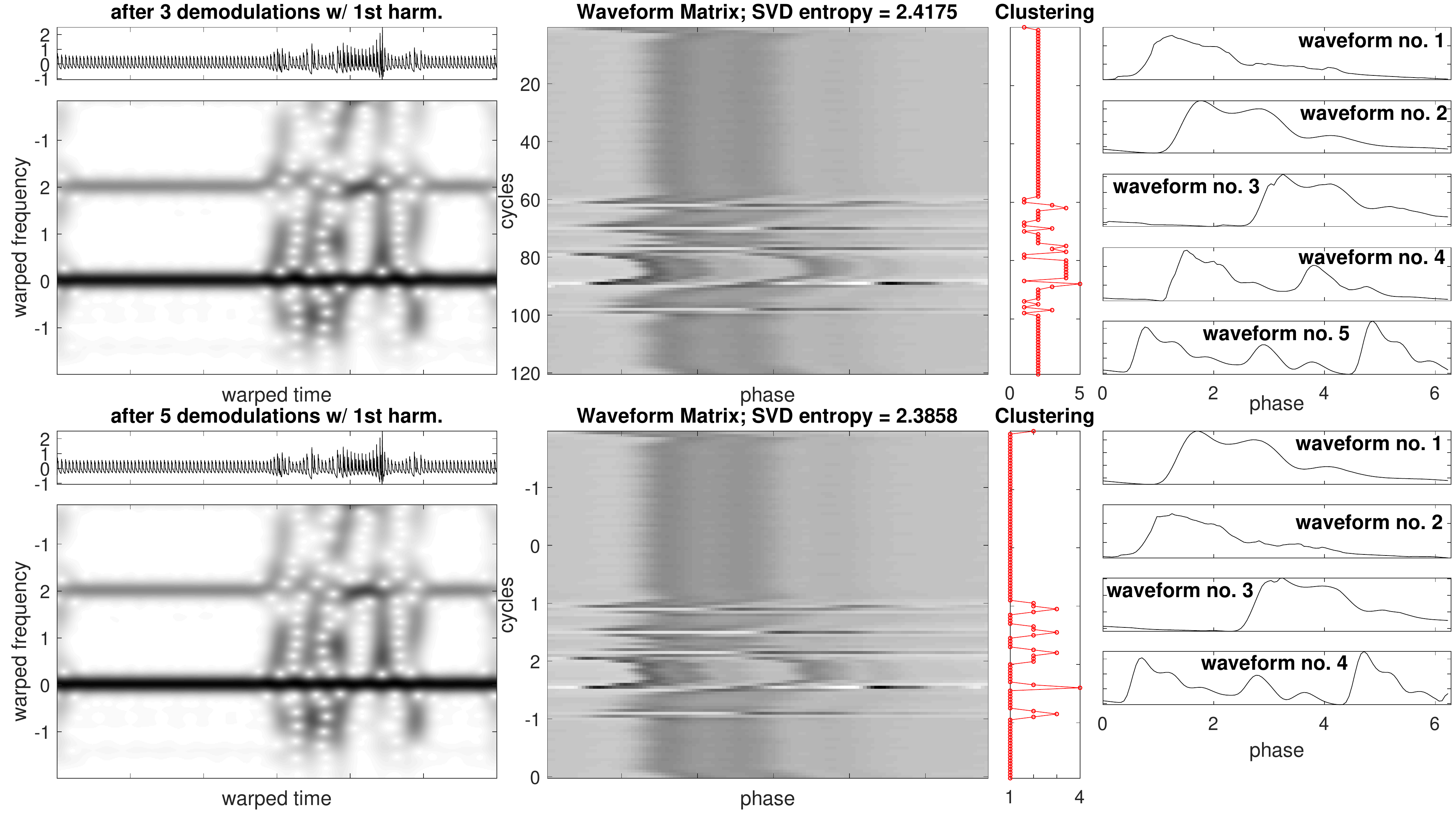}
\includegraphics[width=\columnwidth]{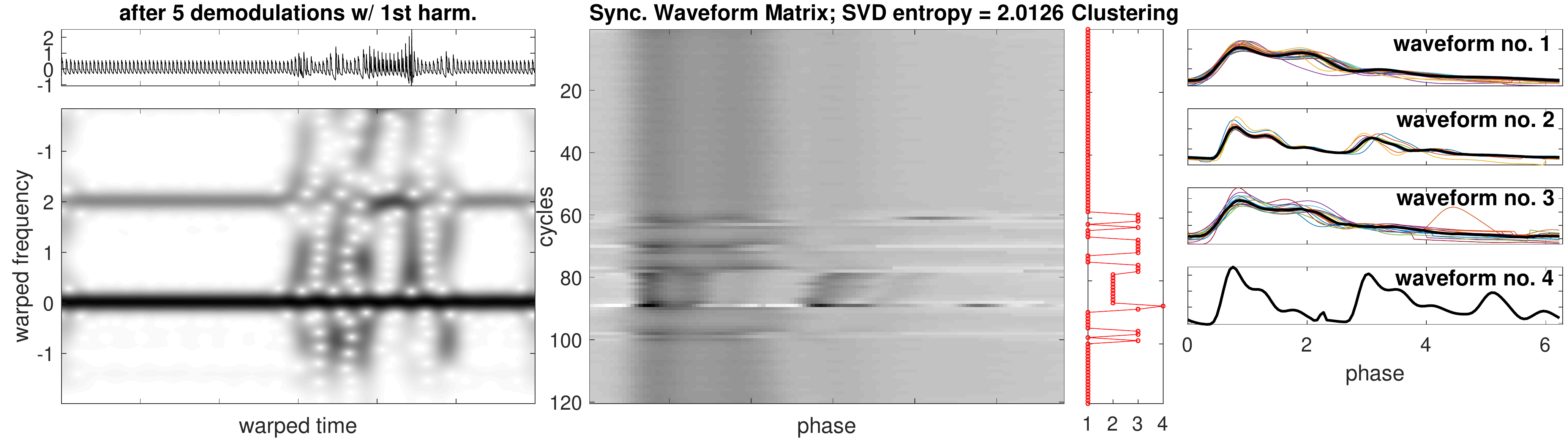}
\end{center}
\caption{An ABP signal with arrhythmia from \cite{Moody2010physionet}. Top row: the original ABP and ECG signals and the spectrogram of the ABP signal with the detected ridge superimposed as a red curve. Second row, from left to right: the warped and demodulated signal and the associated spectrogram, the data matrix with the extracted cycles (the darker the gray color, the larger the value), and the clustering results and the associated estimated WSFs. Third row: same as the second row but after three iterations. Fourth row: same as before but after five iterations. Fifth row: same as the fourth row but with the synchronization.}
\label{fig:ABP}
\end{figure}

\begin{figure}[h]
\begin{center}
\includegraphics[width=\columnwidth]{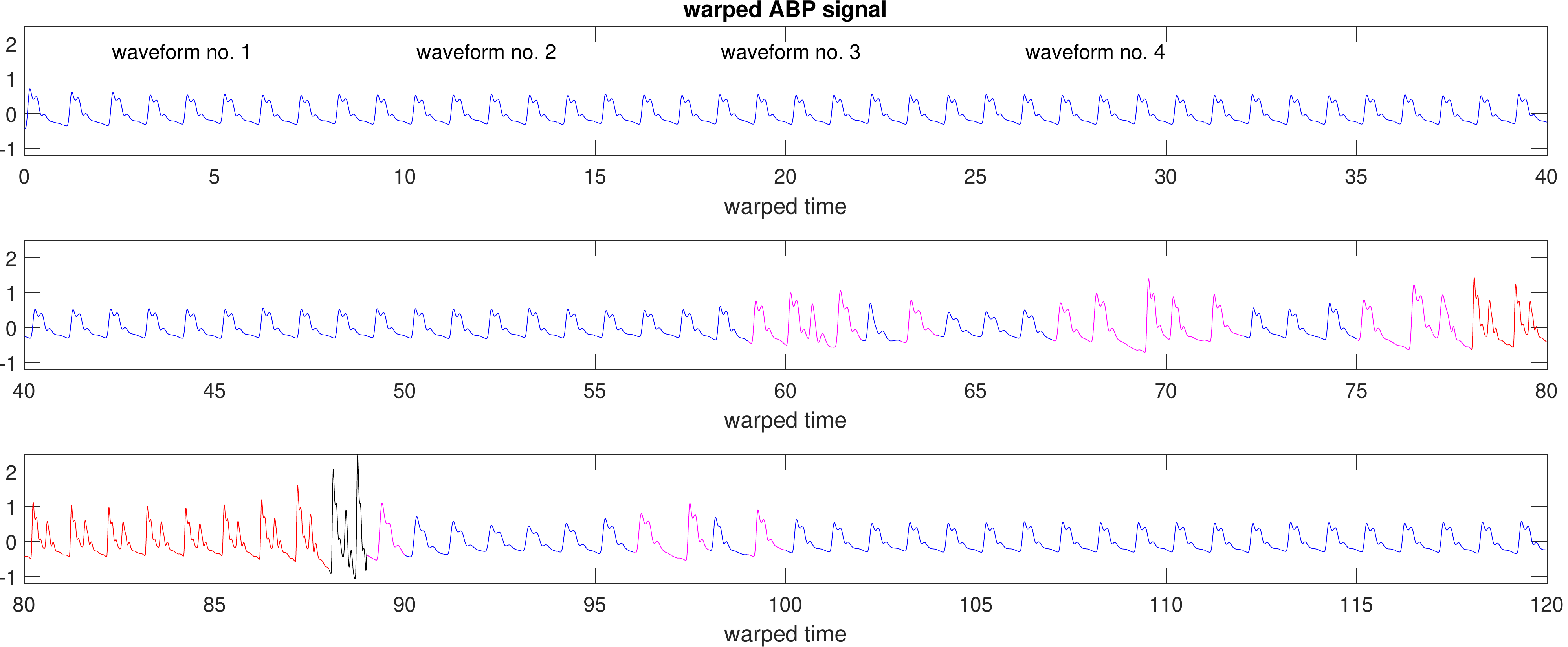}
\end{center}
\caption{The warped ABP signal with arrhythmia from \cite{Moody2010physionet} with cycles colored according the obtained clustering.}
\label{fig:ABP_WARPED}
\end{figure}

Overall, this procedure might serve as an arrhythmia detector.
To demonstrate this potential, Fig. \ref{fig:ABP_WARPED} shows the warped ABP signal, which is colored according to the clustering of the synchronized data matrix after five iterations. The detected WSFs are in agreement with the arrhythmia determined from the simultaneously recorded ECG. The `normal' WSF (the most frequent one) corresponds to the waveform no. 1. The `double-peak' WSF caused by the arrhythmia is caught by the waveform no. 2. An atypical WSF (appearing only once) is represented by the waveform no. 4, and the waveform no. 3 represents the remaining WSFs.

\subsection{Electrocardiogram}
Next, we consider an ECG signal (recording 213) from the MIT-BIH Arrhythmia Database \cite{MIT-BIH}, which presents some premature atrial  beats, ventricular fusion beats and  premature ventricular beats.

An interesting fact about ECG is its spectral distribution. In the spectral domain, the energy of the higher harmonics is relatively high, since visually it ``looks like'' the differentiation of a delta measure convolved with a Gaussian function. Usually, the second harmonic is more dominant than the fundamental component. See an example and its spectrogram shown in the first row of Fig. \ref{fig:ECG}, where the detected ridge for the fundamental component is shown in the red color. Note that in a noisy context, the dominant harmonic will have a better SNR, and hence the phase can be better estimated. On the other hand, we would speculate that the harmonic with most energy carries the most information of the signal.
Thus, we would consider the second harmonic, that is, $\ell=2$, in the algorithm.

The warped and demodulated ECG with the first harmonic and its spectrogram are shown in the second row of Fig. \ref{fig:ECG}. We could see that the IF of the fundamental component has been ``warped'' to be close to a flat line, but the IFs of other harmonics fluctuate ``significantly''. In this setup, the algorithm finds 6 clusters (we only show 4 dominating estimated WSFs), and some estimated WSFs are not physiological.

The situation is rather different when we warp and demodulate with the second harmonic (third row). The spectrogram of the warped and demodulated signal is ``better'' in the sense that the IFs of all harmonics are almost flat. However, since the segmentation is impacted by the imperfect phase estimation, the algorithm still finds up to 6 clusters. We can observe small shifts from one row to the other.
The fourth row shows the results of warping and demodulation with the second harmonic, where we divided the warped time by 1.991. Here, 1.991 is chosen empirically by searching the value between 1.99 and 2.01 that minimizes the SVD entropy. This gives us a better segmentation of the signal in the sense that we have 3 WSFs with physiological meanings. A synchronization of this matrix leads to a satisfactory result, where the SVD entropy is further decreased. The clustering of this matrix retrieves only two WSFs, which correspond to normal beats and PVCs.
Finally, we present in Fig. \ref{fig:ECG_PREM_BEATS} the original version of the ECG signal, with those cycles belonging to the cluster no. 2 in red. This suggests the potential of our algorithm to identify PVCs.

\begin{figure}[bht!]
\begin{center}
\includegraphics[width=\columnwidth]{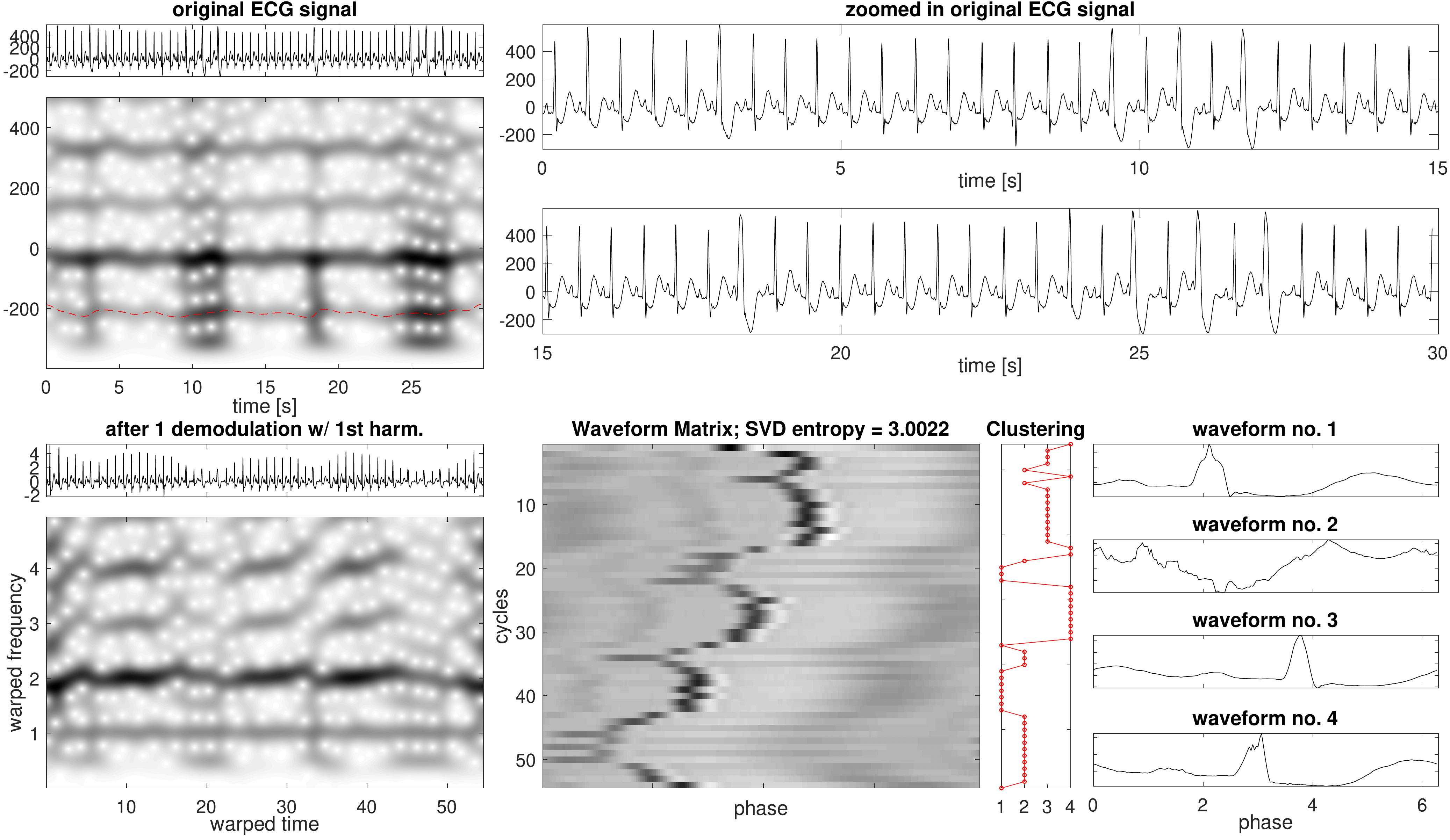}
\includegraphics[width=\columnwidth]{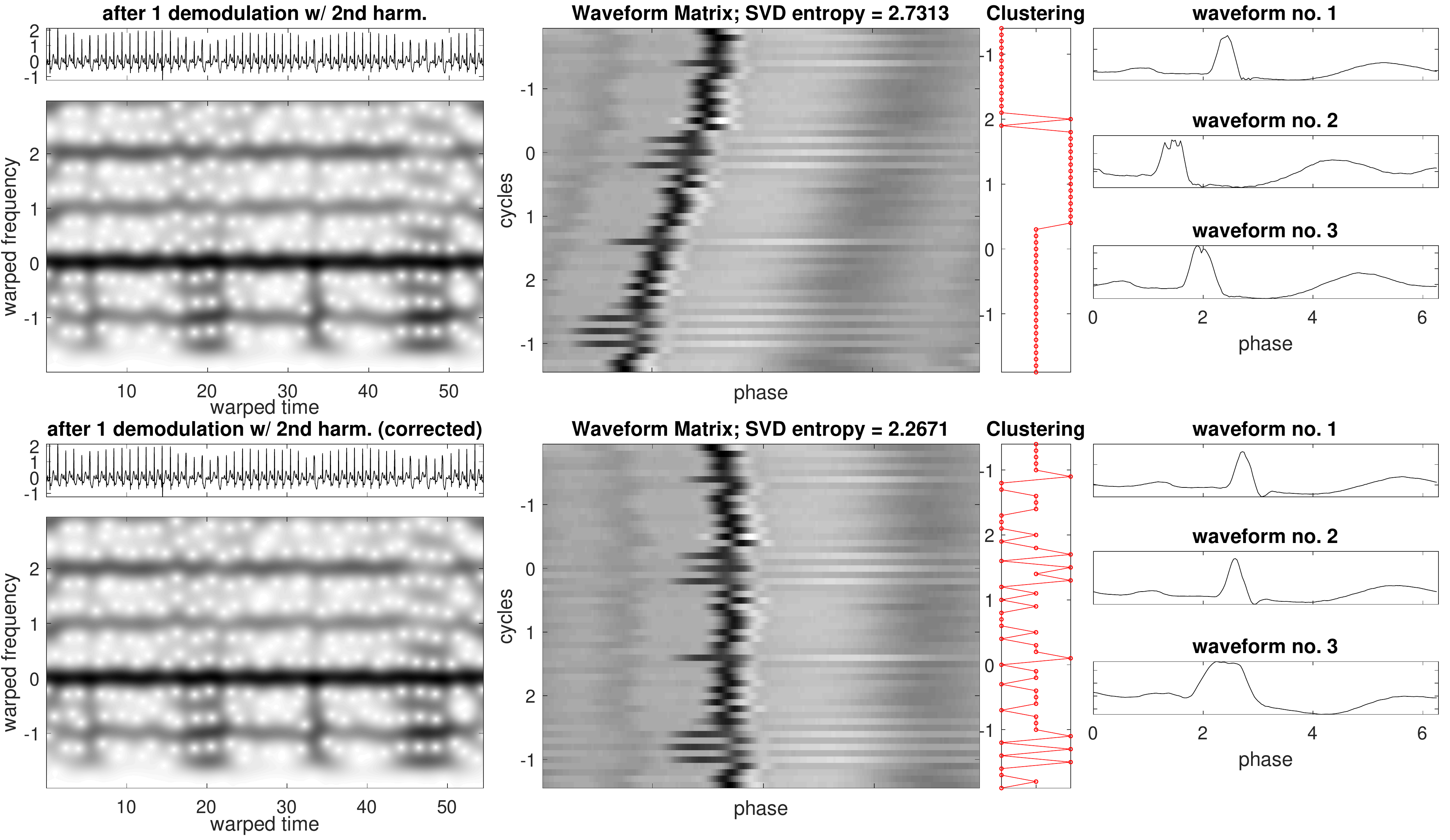}
\includegraphics[width=\columnwidth]{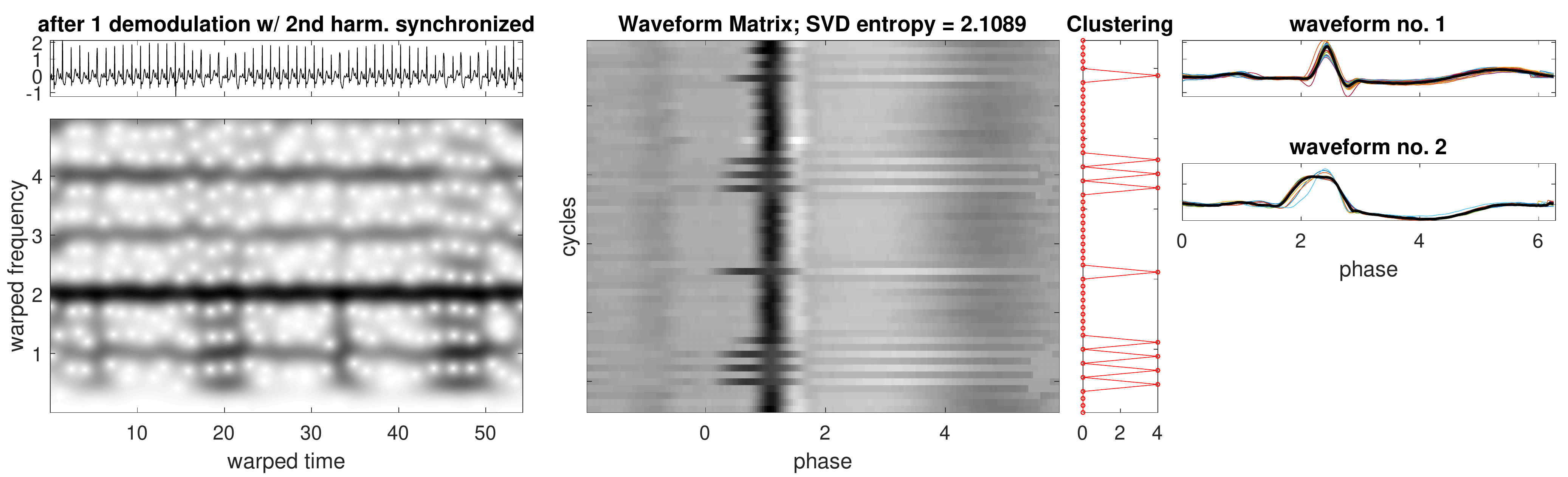}
\end{center}
\caption{An ECG with arrhythmia from \cite{MIT-BIH}. Top: original signal, the spectrogram and detected ridges (red and blue). Second row, from left to right: warped signal and modulus of its STFT; extracted cycles from the once warped ECG (the darker the gray the larger the value); results of the clustering; estimated WSFs (clusters medians). Third row: same as the second row but warped and demodulated with the second harmonic. Fourth row: same as before, but correcting the segmentation. Fifth row: same as before, but the data matrix is synchronized.}
\label{fig:ECG}
\end{figure}

\begin{figure}[bht!]
\begin{center}
\includegraphics[width=\columnwidth]{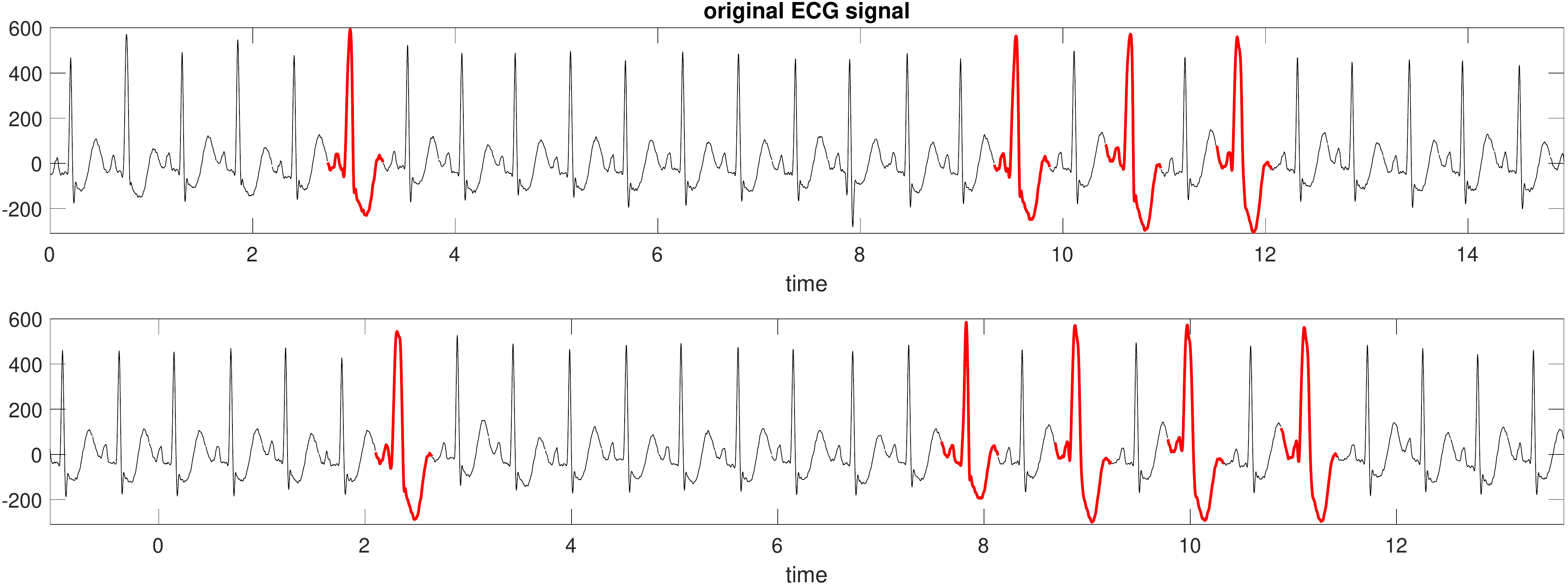}
\end{center}
\caption{ECG signal from \cite{MIT-BIH} with PVCs marked in red.}
\label{fig:ECG_PREM_BEATS}
\end{figure}

\section{Accelerometry data}
The final example is the accelerometer signal from the Indiana University Walking and Driving Study (IUWDS) \cite{strkaczkiewicz2016automatic,fadel2019differentiating}, which is available at \url{https://physionet.org/content/accelerometry-walk-climb-drive/1.0.0/}. The results are presented in the Supplemental Material.

The database contains raw accelerometry data collected from four body locations (left wrist, left hip, left and right ankle), from healthy subjects while walking on the flat floor, ascending and descending stairs, and driving a vehicle.
Following \cite{fadel2019differentiating}, we consider the right ankle \emph{acceleration vector magnitude} signal, which is defined as $x(t) = (a^2(t)+b^2(t)+c^2(t))^{1/2}$, where $a,b$ and $c$ stand, respectively, for the acceleration on the $x$, $y$ and $z$ axes. We chose a segment such that the subject was first walking, and then ascending and descending the stairs. Below we show how our algorithm is able to detect this three different tasks, in a completely unsupervised manner.

As with the previous ECG example, we take $\ell=2$. The original signal, the spectrogram, and three zoomed in segments are shown on the first row of Fig. \ref{fig:IMU}. The second row shows the warped signal along with its spectrogram and the clustering results, where only two WSFs are detected. Increasing the number of iterations leads to a better result with a decreasing SVD entropy. The synchronized matrix after 8 iterations leads to 3 different WSFs reflecting three different tasks. The warped signal, colored according to its labels, is presented on Fig. \ref{fig:IMU_warped}. The WSFs corresponding to walking (in the red color) are perfectly detected, and one cycle from the ascending stairs activity is mislabeled as the descending activity. This mentioned cycle is evidently different from the rest of its cluster, so the clustering result is not surprising. The WSFs corresponding to descending stairs are shown in the magenta color.

\begin{figure}[t]
\begin{center}
\includegraphics[width=\columnwidth]{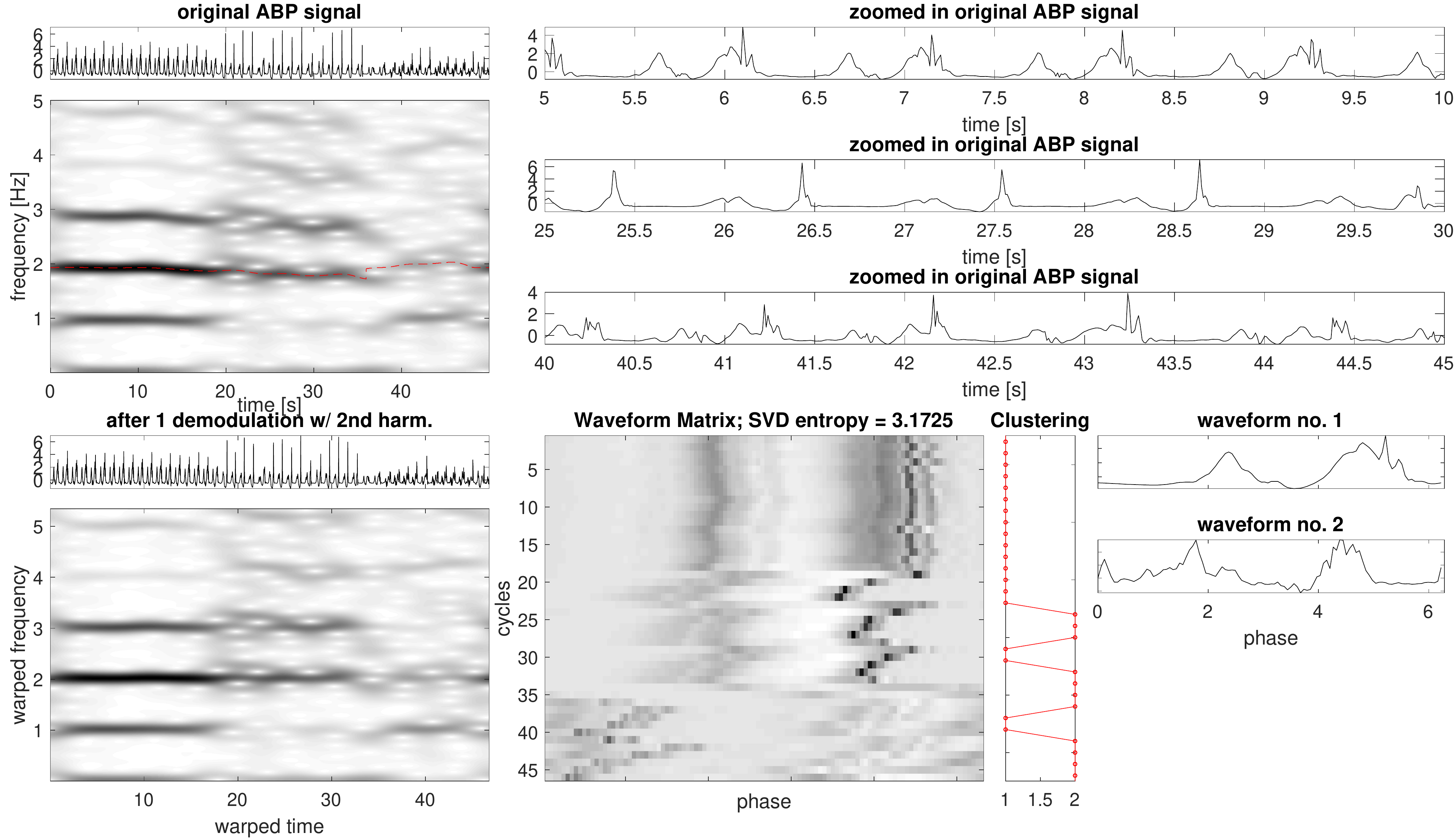}
\includegraphics[width=\columnwidth]{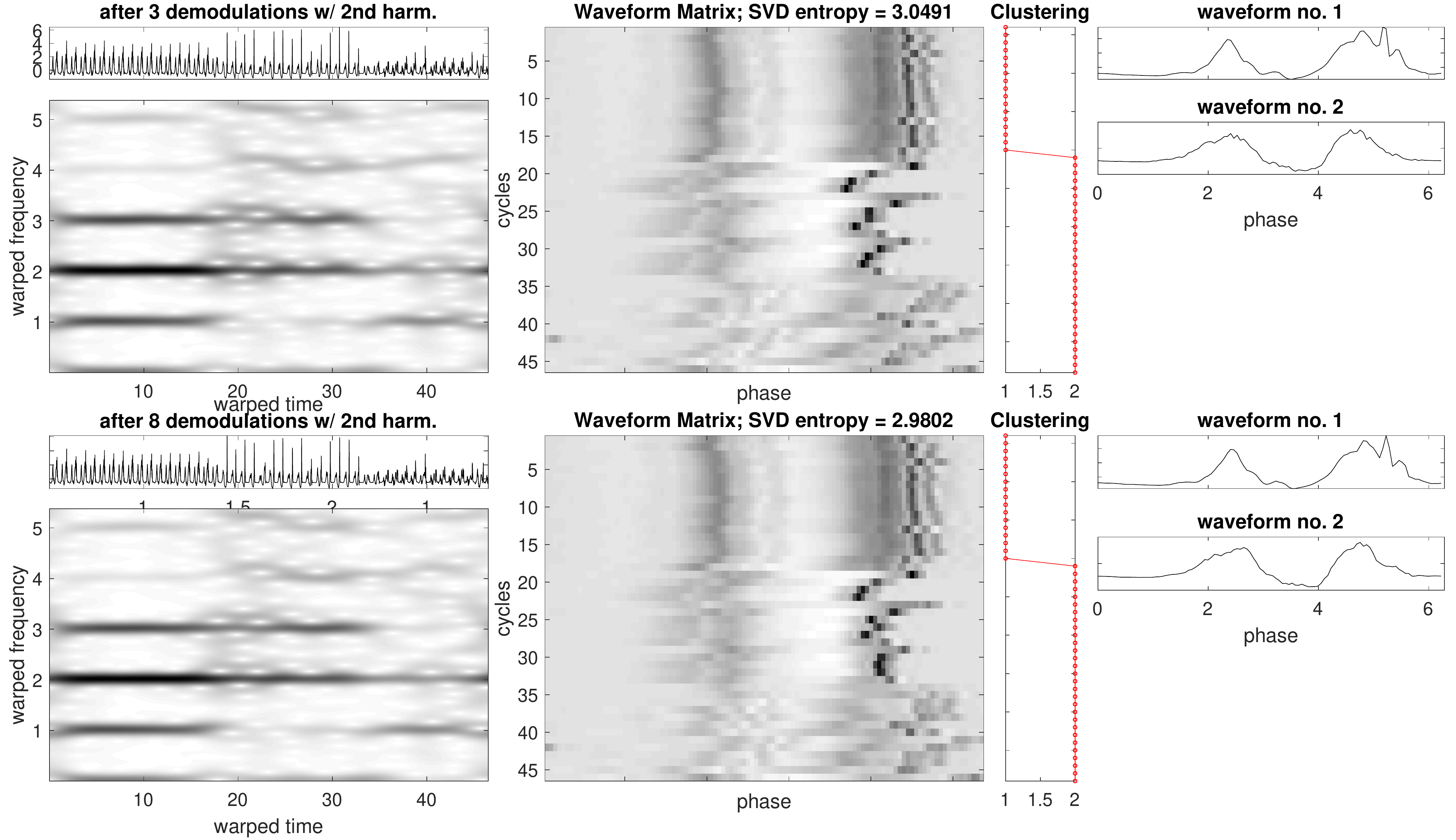}
\includegraphics[width=\columnwidth]{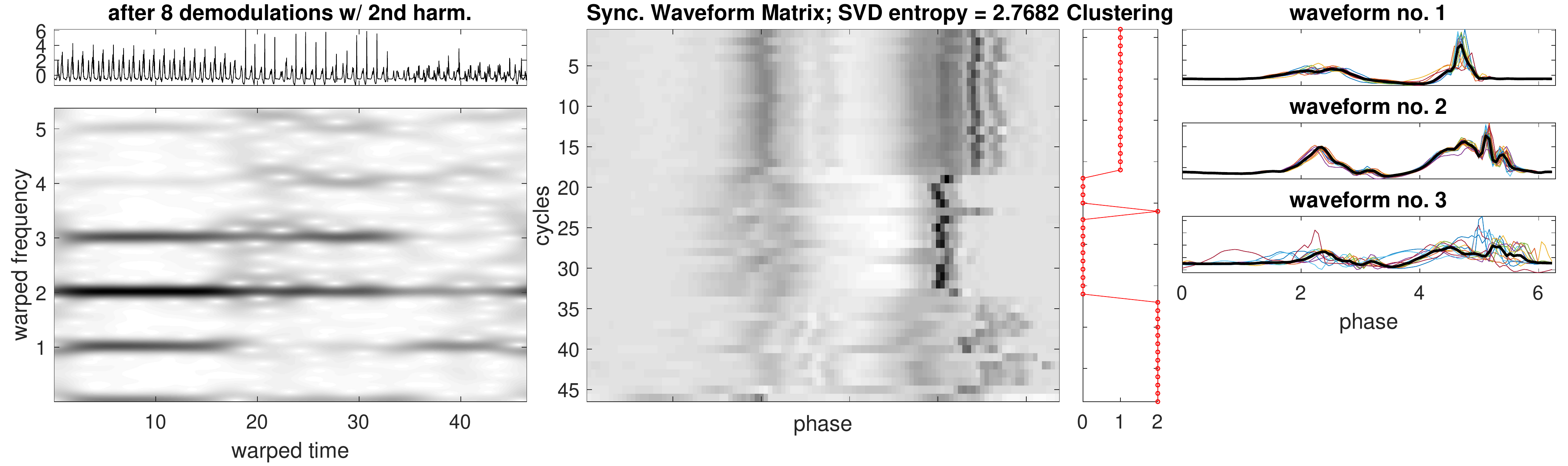}
\end{center}
\caption{Accelerometry from \cite{strkaczkiewicz2016automatic}. Top: original signal, the spectrogram and the detected ridge of the second harmonic (red). Second row, from left to right: warped signal and the spectrogram, the extracted cycles from the warped and demodulated accelerometer signal (darker pixels indicate larger values), and the estimated WSFs. Third row: same as the second row but with 3 iterations with the 2nd harmonic. Fourth row: same as the second row but with 8 iterations. Fifth row: same as before, but the data matrix is synchronized.}
\label{fig:IMU}
\end{figure}

\begin{figure}[t]
\begin{center}
\includegraphics[trim=0 2 0 11, clip,width=\columnwidth]{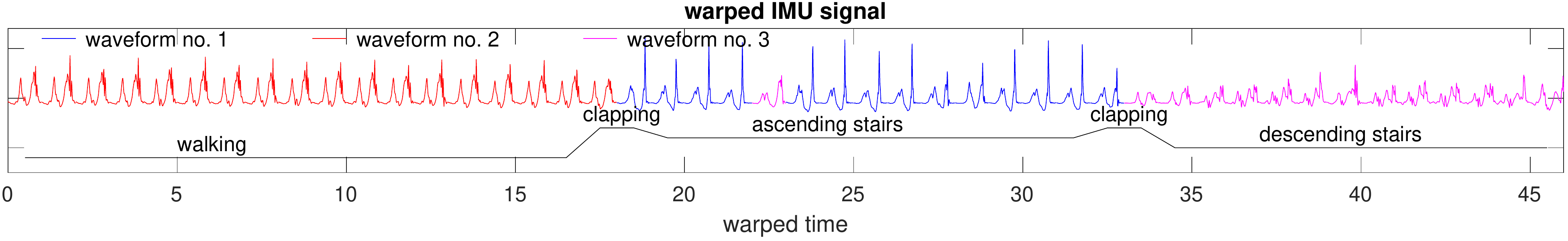}
\end{center}
\caption{Warped accelerometry data with different WSFs shown in different colors. The ground truth is shown, with clapping activity to mark transitions.}
\label{fig:IMU_warped}
\end{figure}

\section{Theoretical Support}\label{sec:theory}

In this section, we provide a theoretical guarantee of the performance of the proposed iterative algorithm. The proof depends on a series of results that might be interesting on its own. To this end, a critical step is spelling out the constants hidden in the asymptotic approximation error terms.
We need the following lemma.

\begin{lemma}\label{lemma h truncation bound}
Let $h(t)$ be a nonzero, real and even Schwartz function such that the window $h$ satisfies $\mbox{supp}\{\hat{h}\} \subseteq[-H,H]$ for $H>0$. For $\beta\geq 0$ and $\xi\in \mathbb{R}$, we can take a large $p\in \mathbb{N}$ so that
\[
\left|\int_{\beta}^\infty  h(x)e^{-i2\pi\xi x}dx\right|\leq  \frac{\sqrt{H}C_{p}}{\max\{\xi,H\}(1+H\beta)^{p-1}}\,,
\]
where $C_{p}>0$ depends on $p$ and $h_0$, which is defined as $h_0(t):=h(t/H)/\sqrt{H}$.
\end{lemma}
\begin{proof}
To standardize the impact of $H$, consider $h_0(t):=h(t/H)/\sqrt{H}$ so that $\|h_0\|_2=\|h\|_2$. Clearly, $h_0$ is a real and even Schwartz function such that $\mbox{supp}\{\widehat{h_0}\}\subset [-1,1]$. By the integration by parts, we have
\begin{align*}
\int_{\beta}^\infty h(x)e^{-i2\pi\xi x}dx=&\frac{i}{2\pi\xi}\left[-\sqrt{H}h_0(H\beta)e^{-2\pi i \xi \beta}\right. \\
&\left.-\int_\beta^\infty \sqrt{H}Hh_0'(Hx)e^{-i2\pi \xi x}dx\right]\,,
\end{align*}
which by the definition of a Schwartz function is bounded by
$\frac{\sqrt{H}C'_{p}}{\xi(1+H\beta)^{p-1}}$,
where $p$ is any large integer and $C'_{p}$ depends on $p$ and $h_0$. Indeed, for any large $p$, there exist constants $K_p>0$ and $K_p'>0$ so that $|h_0(t)|\leq K_p/(1+|t|)^p$ and $|h'_0(t)|\leq K'_p/(1+|t|)^p$. On the other hand, when $|\xi|$ is small, we would bound $\left|\int_{\beta}^\infty h(x)e^{-i2\pi\xi x}dx\right|$ directly by $\sqrt{H}\int_{\beta}^\infty |h_0(Hx)|dx$, which is bounded by $\frac{\sqrt{H}C''_{p}}{(1+H\beta)^{p-1}}$, where $C''_{p}$ depends on $p$ and $h_0$. Thus, we achieve the claimed bound with $C_{p}=\max\{C'_{p},C''_{p}\}$.

\end{proof}

The following lemma is a direct extension of those shown in \cite{Daubechies2011synchro}.

\begin{lemma}\label{lemma: main taylor expansion}
Let $x(t)$ be a signal satisfying \texttt{WSFv3} with the $\ell$-th harmonic being the dominant harmonic and $r\geq 0$. Let $h(t)$ be a real and even Schwartz function such that the window $h$ satisfies $\mbox{supp}\{\hat{h}\} \subseteq[-H,H]$ for $H>0$, and denote $m_k:=\int |h(t)||t|^k dt$ and $m'_k:=\int |h'(t)||t|^k dt$ for $k=0,1,2,\ldots$. Suppose $\phi'_1(t)>H$ and $\phi'_{j+1}-\phi'_j>2H$, for $j=1,2,\ldots$.
Then, we have
\begin{align*}
V_{x}^{(h)}(t,\xi)&\,=\frac{1}{2}\sum_{l=1}^\infty B_l(t)e^{i2\pi(\phi_l(t))}\hat{h}(\xi- \phi'_l(t)) \\
&+C_0(t)\epsilon +E_0(t)\\
\partial_\xi V^{(h)}_x(t,\xi)&\,=\frac{1}{2}\sum_{l=1}^\infty B_l(t)e^{i2\pi(\phi_l(t))}\hat{h}'(\xi- \phi'_l(t))\\
&+C_1(t)\epsilon+E_1(t)\\
\partial_t V^{(h)}_x(t,\xi)&\,=-i\pi\sum_{l=1}^\infty B_l(t)e^{i2\pi(\phi_l(t)-\xi t)}(\xi- \phi'_l(t))\times\\&
\hat{h}(\xi- \phi'_l(t))+C_0(t) \epsilon \xi +C_2(t)\epsilon+E_2(t)\,,
\end{align*}

\noindent where $C_0(t)$, $C_1(t)$ and $C_2(t)$ are complex valued functions that depend on the time-varying IF and AM and satisfy
\begin{align*}
|C_0(t)|&\,\leq \frac{1}{\ell}\left(\phi'_\ell(t)m_1+\frac{Mm_2}{2}\right) \sum_{l = 1}^{\infty} c(l)\\
&+\frac{\pi}{\ell}  \left(\phi'_\ell(t)m_2+\frac{Mm_3}{3} \right)\sum_{l = 1}^{\infty} l B_l(t)\\
|C_1(t)|&\,\leq \frac{2\pi}{\ell}\left(\phi'_\ell(t)m_2+\frac{Mm_3}{2}\right) \sum_{l = 1}^{\infty} c(l)\\
&+\frac{2\pi^2}{\ell}  \left(\phi'_\ell(t)m_3+\frac{Mm_4}{3} \right)\sum_{l = 1}^{\infty} l B_l(t)\\
|C_2(t)|&\,\leq \frac{1}{\ell}\left(\phi'_\ell(t)m_1'+\frac{Mm_2'}{2}\right) \sum_{l = 1}^{\infty} c(l)\\
&+\frac{\pi}{\ell}  \left(\phi'_\ell(t)m_2'+\frac{Mm_3'}{3} \right)\sum_{l = 1}^{\infty} l B_l(t)
\end{align*}
and $E_0(t)$, $E_1(t)$ and $E_2(t)$ are complex valued functions that depend on the change points and satisfy
\begin{align*}
|E_l(t)|\leq &c_{l}
\sum_{i=1}^r\left[\sum_{l_i} \frac{ \Delta B_{l_i}(t_i)}{\max\{|\xi-\phi_{l_i}'(t)|,H\}}\right.\\
&\left.+\sum_{k_i}  \frac{B_{k_i}(t_i^+)\Delta\phi_{k_i}(t_i)}{\max\{|\xi-\phi_{k_i}'(t)|,H\} }\right]\frac{1}{(1+|t-t_i|)^{p-1}}
\end{align*}
for $l=0,1,2$, a large $p\in \mathbb{N}$, and constants $c_l>0$ that depend on $p$ and $h_0$.
\end{lemma}

\begin{remark}
The assumption $\mbox{supp}\{\hat{h}\} \subseteq[-H,H]$ for $H>0$ could be easily relaxed since the tail of a Schwartz function is light. The proof would be the same with more tedious notations to handle the light tail, while the results will not provide further insights into the problem. Clearly, when there is no change point, that is, when $r=0$, $E_l(t)=0$. Note that when $t$ is far away from the change points, or when $\xi$ is far away from the IF of the harmonics with change points, $E_l(t)$ is small. In particular, over the dominant harmonic, the impact of the change point is well controlled. The same remark holds for this whole section.
\end{remark}

\begin{proof}
First, assume $r=0$.
Fix time $t_0$. We start with claiming the following two bounds
\begin{align*}
|B_l(t)-B_l(t_0)|&\, \leq
\epsilon \frac{c(l)}{\ell} |t-t_0|\left(\phi'_\ell(t_0)+\frac{1}{2}M|t-t_0|\right)\\
|\phi'_l(t)-\phi'_l(t_0)| &\,\leq
\epsilon \frac{l}{\ell}|t-t_0|\left(\phi'_\ell(t_0)+\frac{1}{2}M|t-t_0|\right) \,.
\end{align*}
Indeed, we have
 \begin{align*}
    \left|B_l(t)-B_l(t_0)\right|\,
    & =  \left| \int_0^{t-t_0}B'_l(t_0+u)du\right|  \\
    &\leq   \int_0^{t-t_0} \left|B'_l(t_0+u)\right|d u\\
     &\leq \epsilon  \frac{c(l)}{\ell} \int_0^{t-t_0} \left| \phi'_\ell(t_0+u)\right|  d u\\
    & = \epsilon  \frac{c(l)}{\ell} \int_0^{t-t_0}\left|
      \phi'_\ell(t_0) + \int_0^u\phi''_\ell(t_0+x) dx\right|  d u \\
    & \leq \epsilon \frac{c(l)}{\ell} \left(
        \phi'_\ell(t_0)|t-t_0|+\frac{1}{2}M |t-t_0|^2\right)\,,
  \end{align*}
and the other bound comes from the same calculation.
Denote
\[
x_{t_0}(t):=\sum_{l = 1}^{\infty} B_l(t_0) \cos(2 \pi (\phi_l(t_0)+\phi_l'(t_0)(t-t_0)))
\]
and
\[
E_{t_0}(t):=x(t)-x_{t_0}(t)\,.
\]
Clearly, $E_{t_0}(t)$ is a smooth function.
We prepare the following bound:
\begin{equation}\label{trivial cos bound}
\begin{aligned}
|\cos(&2 \pi (\phi_l(t_0)+\phi_l'(t_0)(t-t_0)))-\cos(2 \pi \phi_l(t))|\\
&\leq \frac{\pi l}{\ell} \epsilon \left(\phi'_\ell(t_0)|t-t_0|^2+\frac{M}{3}|t-t_0|^3 \right)\,.
\end{aligned}
\end{equation}
To show it, without loss of generality, assume $t>t_0$. We have
\begin{align*}
|\cos(&2 \pi (\phi_l(t_0)+\phi_l'(t_0)(t-t_0)))-\cos(2 \pi \phi_l(t))|\\
&\leq\, 2\pi|\phi_l(t_0)+\phi_l'(t_0)(t-t_0)-\phi_l(t)|\\
&\leq\, 2\pi\int_{0}^{t-t_0}|\phi_l'(t_0+s)-\phi_l'(t_0)|ds\\
&\leq \frac{2\pi l}{\ell} \epsilon \int_{0}^{t-t_0}  s\left(\phi'_\ell(t_0)+\frac{1}{2}Ms\right) ds\\
&=\,\frac{\pi l}{\ell} \epsilon \left(\phi'_\ell(t_0)(t-t_0)^2+\frac{M}{3}(t-t_0)^3 \right)\,,
\end{align*}
and hence the claim.
Thus,
by a direct bound, we have the following control of $E_{t_0}(t)$:
\begin{align*}
|E_{t_0}(t)|&=\left|\sum_{l = 1}^{\infty} B_l(t) \cos(2 \pi \phi_l(t))\right.\\
-&\left.\sum_{l = 1}^{\infty} B_l(t_0) \cos(2 \pi (\phi_l(t_0)+\phi_l'(t_0)(t-t_0)))\right|\\
&\leq \sum_{l = 1}^{\infty} |B_l(t)-B_l(t_0)| |\cos(2 \pi \phi_l(t))|\\
+&\sum_{l = 1}^{\infty} B_l(t_0) |\cos(2 \pi (\phi_l(t_0)+\phi_l'(t_0)(t-t_0)))\\
&\quad\quad-\cos(2 \pi \phi_l(t))|\\
&\leq\frac{\epsilon}{\ell}|t-t_0|\left(\phi'_\ell(t_0)+\frac{1}{2}M|t-t_0|\right)\sum_{l = 1}^{\infty} c(l)\\
+\frac{\pi \epsilon}{\ell} &|t-t_0| \left(\phi'_\ell(t_0)|t-t_0|+\frac{M}{3}|t-t_0|^2 \right) \sum_{l = 1}^{\infty} l B_l(t_0)\,,
\end{align*}
where \eqref{trivial cos bound} is used in the second inequality and both $\sum_{l = 1}^{\infty} c(l)$ and $\sum_{l = 1}^{\infty} l B_l(t_0)$ are finite by Assumption (T1).
To finish the proof, note that due to the support of $\hat{h}$, we have
\[
V_{x_{t_0}}^{(h)}(t_0,\xi)=\frac{1}{2}\sum_{l=1}^\infty B_l(t_0)e^{i2\pi(\phi_l(t_0))}\hat{h}(\xi- \phi'_l(t_0))\,.
\]
Note that $V_x^{(h)}-V_{x_{t_0}}^{(h)}=V_{E_{t_0}}^{(h)}$, which is a smooth function. With the definition of $m_k$, we have an immediate bound of $|V_{E_{t_0}}^{(h)}|$ by
\begin{align*}
&|V_{E_{t_0}}^{(h)}(t_0,\xi)|\leq \int |E_{t_0}(t)||h(t-t_0)|dt\\
\leq \,& \epsilon\left[\frac{1}{\ell}\left(\phi'_\ell(t_0)m_1+\frac{Mm_2}{2}\right) \sum_{l = 1}^{\infty} c(l)\right.\\
&\left.+\frac{\pi}{\ell}  \left(\phi'_\ell(t_0)m_2+\frac{Mm_3}{3} \right)\sum_{l = 1}^{\infty} l B_l(t_0)\right]\,,
\end{align*}
and hence the proof of the first part. For the other parts, note that
\[
\partial_\xi V^{(h)}_x(t,\xi) = -i2\pi V^{(th)}_x(t,\xi)\]
and
\[
\partial_t V^{(h)}_x(t,\xi) = i2\pi \xi V^{(h)}_x(t,\xi)- V^{(h')}_x(t,\xi).
\]
Note that $\widehat{th(t)}(\xi) = \frac{-1}{i2\pi}\hat{h}'(\xi)$ and $\widehat{h'}(\xi)=i2\pi\xi\hat{h}(\xi)$, so we have
\begin{align*}
V_{x_{t_0}}^{(th)}(t_0,\xi)&\,=\frac{-1}{i4\pi}\sum_{l=1}^\infty B_l(t_0)e^{i2\pi(\phi_l(t_0))}\hat{h}'(\xi- \phi'_l(t_0))\\
V_{x_{t_0}}^{(h')}(t_0,\xi)&\,=i\pi\sum_{l=1}^\infty B_l(t_0)e^{i2\pi(\phi_l(t_0))}(\xi- \phi'_l(t))\hat{h}(\xi- \phi'_l(t_0))\,.
\end{align*}
On the other hand, we have
\begin{align*}
&|V_{E_{t_0}}^{(th)}(t_0,\xi)|\leq \int |E_{t_0}(t)||t-t_0||h(t-t_0)|dt\\
\leq \,& \epsilon\left[\frac{1}{\ell}\left(\phi'_\ell(t_0)m_2+\frac{Mm_3}{2}\right) \sum_{l = 1}^{\infty} c(l)\right.\\
&+\left.\frac{\pi}{\ell}  \left(\phi'_\ell(t_0)m_3+\frac{Mm_4}{3} \right)\sum_{l = 1}^{\infty} l B_l(t_0)\right]
\end{align*}
and
\begin{align*}
&|V_{E_{t_0}}^{(h')}(t_0,\xi)|\leq \int |E_{t_0}(t)||h'(t-t_0)|dt\\
\leq \,& \epsilon\left[\frac{1}{\ell}\left(\phi'_\ell(t_0)m_1'+\frac{Mm_2'}{2}\right) \sum_{l = 1}^{\infty} c(l)\right.\\
&\left.+\frac{\pi}{\ell}  \left(\phi'_\ell(t_0)m_2'+\frac{Mm_3'}{3} \right)\sum_{l = 1}^{\infty} l B_l(t_0)\right]\,.
\end{align*}

To finish the proof, we consider the case when $r\geq 1$. Without loss of generality, assume $r=1$, $B_l$ is discontinuous at $t_1$, and $\phi_l$ is discontinuous at $t_1$ for some $l\neq \ell$. Fix $t_0\leq t_1$. By (T4), we know
\[
\tilde{B}_l(t):=B_l(t)-\Delta B_l(t_1)\chi_{[t_1,\infty)}(t)\in C^1(\mathbb{R})\,,
\]
where $\Delta B_l(t_1):=B_l(t_1^+)-B_l(t_1^-)$ and
\[
\tilde{\phi}_l(t):=\phi_l(t)-\Delta \phi_l(t_1)\chi_{[t_1,\infty)}(t)\in C^2(\mathbb{R})\,,
\]
where $\Delta \phi_l(t_1):=\phi_l(t_1^+)-\phi_l(t_1^-)$.
Note that the above result when $r=0$ still holds for $\tilde{x}:=\sum_{j \neq l} B_j(t) \cos(2 \pi \phi_j(t))+\tilde B_l(t) \cos(2 \pi \tilde\phi_l(t))$ at $t_0$. Thus, to obtain the result for $\sum_{j=1}^\infty B_j(t) \cos(2 \pi \phi_j(t))$, we need to control the difference in $V^{(h)}_x$, $\partial_\xi V^{(h)}_x$ and $\partial_t V^{(h)}_x$ caused by $\tilde{B}_l(t)-B_l(t)=-\Delta B_l(t_1)\chi_{[t_1,\infty)}(t)$ and $\tilde{\phi}_l(t)-\phi_l(t)=-\Delta \phi_l(t_1)\chi_{[t_1,\infty)}(t)$ by taking Lemma \ref{lemma h truncation bound} into account. For example, by the same approximation of $\cos(2 \pi \tilde\phi_l(t))$ and the control of $\cos(2 \pi \phi_l(t))-\cos(2 \pi \tilde\phi_l(t))$ as the above, $|V^{(h)}_x-V^{(h)}_{\tilde{x}}|$ is controlled by
\begin{align*}
&\left|\int_{t_1}^\infty [B_l(t_1^+)-B_l(t_1^-)] \cos(2 \pi \tilde\phi_l(t)) h(t-t_0)e^{-i2\pi\xi t}dt\right|\\
&+\left|\int_{t_1}^\infty B_l(t)[\cos(2 \pi \tilde\phi_l(t))-\cos(2 \pi \phi_l(t))] h(t-t_0)e^{-i2\pi\xi t}dt\right|\\
&\leq c_{p}\frac{\Delta B_l(t_1)+B_l(t_1+)\Delta\phi_l(t_1)}{\max\{|\xi-\phi_l'(t_0)|,H\} (1+|t_0-t_1|)^{p-1}}\,,
\end{align*}
where $c_p>0$ is a constant depending on $p$ and $h_0(t)=h(t/H)/\sqrt{H}$. Similar arguments hold for $\partial_\xi V^{(h)}_x$ and $\partial_t V^{(h)}_x$.
\end{proof}

Below, to simplify the heavy notation, we assume $r=0$. When $r>0$, all the error terms will be complicated by including the control of $E_l(t)$ in Lemma \ref{lemma: main taylor expansion}.
The following theorem generalizes the existing result \eqref{Traditional result of STFT}.

\begin{theorem}\label{Theorem Ridge regularity of STFT}
Let $x(t)$ be a signal satisfying \texttt{WSFv3} with the $\ell$-th harmonic the dominant harmonic and $r= 0$. Let $h(t)$ be a real and even Schwartz function such that the window $h$ satisfies $\mbox{supp}\{\hat{h}\} \subseteq[-H,H]$ for $H>0$. Suppose $\phi'_1(t)>H$ and $\phi'_{j+1}-\phi'_j>2H$, for $j=1,2,\ldots$.
Then, when $\epsilon$ is sufficiently small, the ridge of the spectrogram of $x$, $|V_x^{(h)}(t,\xi)|^2$, near $\phi'_\ell(t)$, denoted as $(t,\varpi_\ell(t))$, satisfies
\begin{enumerate}
\item $\varpi_\ell(t)=\phi_\ell'(t)+C_4(t)\epsilon$,
where $C_4(t)$ is a real valued function with bound detailed in \eqref{bound of C_4};
\item $\varpi_\ell(t)$ is $C^1$;
\item $|\varpi'_\ell(t)|\leq C_7(t)\epsilon$,
where $C_7(t)$ is a positive valued function with bound detailed in \eqref{bound of C_7}.
\end{enumerate}
\end{theorem}

\begin{remark}
Note that when $r>0$, the error terms, for example, $C_4(t)\epsilon$ and $C_7(t)\epsilon$ in (1) and (3) in Lemma \ref{lemma: main taylor expansion} respectively, will further depend on $E_l(t)$. This holds since the STFT is still a smooth function, and the same proof holds.
\end{remark}

\begin{proof}
Without loss of generality, assume $\ell=1$.
In the proof we need the following controls, which come immediately from the definition of $h$:
\begin{equation}\label{EQ hat h' hat h hat h'' control}
\begin{aligned}
|\hat{h}'(\xi)|&\leq m_2|\xi| \\
|\hat{h}(\xi)- \hat{h}(0)|&=\frac{m_2}{2}|\xi|^2 \\
|\hat{h}''(\xi)- \hat{h}''(0)|&=\frac{m''_2}{2}|\xi|^2
\end{aligned}
\end{equation}
when $|\xi|$ is sufficiently small. Also, note that $|\hat{h}(0)|$ and $|\hat{h}''(0)|$ are of order $1$.
Denote the spectrogram of $x$ as $S_{x}(t,\xi) = |V^{(h)}_{x}(t,\xi)|^2$.
Since $h$ is real symmetric, by Lemma \ref{lemma: main taylor expansion}, we have
\begin{align*}
&F(t,\xi)=\partial_\xi S_{x}(t,\xi)=2\Re[\partial_\xi V^{(h)}_{x}(t,\xi)\overline{V^{(h)}_{x}(t,\xi)}]\\
=&\, \frac{1}{2}\sum_{l=1}^\infty B_l^2(t)\hat{h}(\xi- \phi'_l(t))\hat{h}'(\xi- \phi'_l(t)) + C_3(t,\xi)\epsilon\,,
\end{align*}
where $C_3(t,\xi)$ is a real valued function satisfying
\begin{align*}
|C_3(t,\xi)|\leq& 2\left(|C_1(t)|\sum_{l=1}^\infty B_l(t)|\hat{h}(\xi- \phi'_l(t))|\right.\\
&\left.+|C_0(t)|\sum_{l=1}^\infty B_l(t)|\hat{h}'(\xi- \phi'_l(t))|\right),
\end{align*}
where $\epsilon>0$ is sufficiently small.
Clearly, $\sum_{l=1}^\infty B_l^2(t)\hat{h}(\xi- \phi'_l(t))\hat{h}'(\xi- \phi'_l(t))=0$ when $\xi=\phi'_1(t)$.
By \cite{delprat1992asymptotic}, we know that the ridge associated with $f_1(t) = B_1(t)\cos(2\pi\phi_1(t))$ at time $t$, denoted as $\xi_t$ so that $F(t,\xi_t)=0$, is close to $\phi_1'(t)$; that is, for a given time $t$, there exist $\xi_t>0$ so that $\xi_t=\phi_1'(t)+O(\epsilon)$. To quality the implied constant, note that when $|\xi-\phi_1'(t)|<H$, we have $F(t,\xi)=\frac{1}{2} B_1^2(t)\hat{h}(\xi- \phi'_1(t))\hat{h}'(\xi- \phi'_1(t)) + C_3(t,\xi)\epsilon$, where
\[
|C_3(t,\xi)|\leq 2\left(|C_1(t)| +|C_0(t)| \right)B_1(t)\,.
\]
Thus, by a direct perturbation,
\begin{align}\label{EQ xi_x=phi_1'(x)+O(epsilon)}
\xi_t=\phi_1'(t)+C_4(t)\epsilon\,,
\end{align}
where $C_4(t)$ is a real valued function satisfying
\begin{align}\label{bound of C_4}
|C_4(t)|\leq \frac{8(|C_1(t)| +|C_0(t)| )}{B_1(t)m_2}\,.
\end{align}
Then, we apply the implicit function theory to $F$. To this end, we show that
\[
\partial_\xi F(t,\xi_t)\neq 0\,.
\]
Indeed, we have
\begin{align*}
\partial&_\xi F(t,\xi_t)\\
&=-8\pi^2\big[\Re[V_x^{(t^2h)}(t,\xi_t)\overline{V_x^{(h)}(t,\xi_t)}]-2|V_x^{(th)}(t,\xi_t)|^2\big]\,,
\end{align*}
where $\Re$ means taking the real part.
Note that by Lemma \ref{lemma: main taylor expansion}, we have
$V_{f_1}^{(h)}(t,\xi_t) \approx \frac{B_1(t)}{2} e^{i2\pi\phi_1(t)} \hat{h}(\xi_t- \phi'_1(t))$, $V_{f_1}^{(t^2h)}(t,\xi_t) \approx \frac{-1}{4\pi^2}\frac{B_1(t)}{2} e^{i2\pi\phi_1(t)} \hat{h}''(\xi_t- \phi'_1(t))$, and $V_{f_1}^{(th)}(t,\xi_t) \approx \frac{1}{-i2\pi}\frac{B_1(t)}{2} e^{i2\pi\phi_1(t)} \hat{h}'(\xi_t- \phi'_1(t))$.
As a result,
\begin{align}\label{eq Vh*Vt2h}
&V_{f_1}^{(h)}(t,\xi_t) \overline{V_{f_1}^{(t^2h)}(t,\xi_t)}
\approx \frac{B_1(t)^2}{16\pi^2}  \hat{h}(\xi_t- \phi'_1(t)) \overline{\hat{h}''(\xi_t- \phi'_1(t))}
\end{align}
is of order $1$, and
\begin{align}\label{eq Vh*Vt2h2}
|V_{f_1}^{(th)}(t,\xi_t)|^2\approx\frac{B_1(t)^2}{16\pi^2} |\hat{h}'(\xi_t- \phi'_1(t))|^2=O(\epsilon^2)\,,
\end{align}
where the last equality comes from \eqref{EQ xi_x=phi_1'(x)+O(epsilon)}, which combined lead to the claim.
Thus, by the implicit function theory, we know there exists an open set $U$ around $t$ and a $C^1$ function $\varpi$ on $U$ so that $F(s, \varpi(s))=0$ for $s\in U$. To control $\varpi'(t)$, note that by the implicit function theorem, we have
\[
\varpi'(s)=\frac{\partial_t F(s,\varpi(s))}{|\partial_\xi F(s,\varpi(s))|}\,.
\]
To finish the claim, we need bounds for $|\partial_t F(t,\varpi(t))|$ and $|\partial_\xi F(t,\varpi(t))|$. First,
\begin{align*}
\partial_t F(t,\varpi(t))= -4\pi\Im [V_{f_1}^{(th')}(t,\varpi(t))\overline{V_{f_1}^{(h)}(t,\varpi(t))}] \,,
\end{align*}
where $\Im$ means taking the imaginary part and $th'$ inside the parenthesis means the window $[th]'(t)$.
By the same argument as that in Lemma \ref{lemma: main taylor expansion}, we have
\begin{align*}
V&_{f_1}^{(th')}(t,\varpi(t)) = \\
&-\frac{B_1(t)}{2} e^{i2\pi\phi_1(t)}(\varpi(t)- \phi'_{1}(t)) \hat{h}'(\varpi(t)- \phi'_{1}(t))+C_5(t)\epsilon\,,
\end{align*}
where $C_5(t)$ is a complex valued function satisfying
\begin{align*}
|C_5(t)|&\leq \frac{1}{\ell}\left(\phi'_\ell(t)m_2'+\frac{Mm_3'}{2}\right) \sum_{l = 1}^{\infty} c(l)\\
+&\frac{\pi}{\ell}  \left(\phi'_\ell(t)m_3'+\frac{Mm_4'}{3} \right)\sum_{l = 1}^{\infty} l B_l(t) \,.
\end{align*}
On the other hand, since the Fourier transform of $[th]'(t)$ is $-\xi\hat{h}'(\xi)$, which is an even function, and $\widehat{[th]'}(\xi)=O(|\xi|)$ when $|\xi|\ll 1$, we have
\[
\frac{B_1(t)}{2} e^{i2\pi\phi_1(t)}(\varpi(t)- \phi'_{1}(t)) \hat{h}'(\varpi(t)- \phi'_{1}(t))=O(\epsilon^2)
\]
since $|\varpi(t)- \phi'_{1}(t)|=O(\epsilon)$ and $|\hat{h}'(\varpi(t)- \phi'_{1}(t))|=O(\epsilon)$.
By combining the above bounds, we have
\[
\partial_t F(t,\varpi(t))=C_6(t)\epsilon\,,
\]
where $C_6(t)$ is a complex valued function satisfying
\[
|C_6(t)|\leq 3\pi B_1(t) |C_5(t)|
\]
when $\epsilon$ is sufficiently small.
Finally, we show that $|\partial_\xi F(x,\varpi(x))|$ is of order $1$. Indeed, we have
\begin{align*}
\partial&_\xi F(t,\varpi(t))\\
&= -8\pi^2(\Re [V_{f_1}^{(t^2h)}(t,\varpi(t))\overline{V_{f_1}^{(h)}(t,\varpi(t))}] -|V_{f_1}^{(th)}(t,\varpi(t))|^2) \,,
\end{align*}
where by \eqref{eq Vh*Vt2h} and \eqref{eq Vh*Vt2h2} we know
\[
|\partial_\xi F(t,\varpi(t))| \geq   \frac{B_1(t)^2}{8}  \,.
\]
We thus finish the claim that
\[
|\varpi'(t)|\leq C_7(t)\epsilon\,,
\]
where $C_7(t)$ is a positive valued function satisfying
\begin{align}\label{bound of C_7}
C_7(t)\leq \frac{24\pi }{B_1(t)} |C_5(t)|\,.
\end{align}
\end{proof}

With Theorem \ref{Theorem Ridge regularity of STFT}, we immediately have the following theorem that guarantees the performance of warping and demodulation.

\begin{theorem}\label{Theorem_warping}
Following assumptions stated in Theorem \ref{Theorem Ridge regularity of STFT}, if we denote $\theta_\ell(t)=\phi_\ell(\tilde{\phi}_\ell^{-1}(t))$, where $\tilde{\phi}_\ell$ is the estimated phase from the ridge, we have
\begin{align*}
\theta_\ell(t)&=t+C_{10}(t)\epsilon\\
\theta_\ell'(t)&=1+C_{11}(\tilde{\phi}_\ell^{-1}(t))\epsilon\\
|\theta_\ell''(t)|&\leq \frac{2(H+|C_7(\tilde{\phi}_\ell^{-1}(t))|)}{H^2}\theta_\ell'(t)\epsilon
\end{align*}
for all $t\in \mathbb{R}$, where $C_{10}(t)$ and $C_{11}(t)$ are defined in the proof.
If we denote $A_\ell(t) = \frac{B_\ell(\tilde{\phi}_\ell^{-1}(t))}{\tilde B_\ell(\tilde{\phi}_\ell^{-1}(t))}$, where $\tilde B_\ell$ is the estimated amplitude of the $\ell$-th harmonic from the ridge, we have
\[
A_\ell(t)=1+C_{12}(\tilde{\phi}_\ell^{-1}(t))\epsilon\ \ \mbox{and}\ \ |A'_\ell(t)|\leq C_{14}(\tilde{\phi}_\ell^{-1}(t))\epsilon
\]
for all $t\in \mathbb{R}$, where $C_{12}(t)$ and $C_{14}(t)$ are defined in \eqref{definition C12} and \eqref{definition C14} respectively.
\end{theorem}

\begin{remark}
When $r>0$, the same remark after Theorem \ref{Theorem Ridge regularity of STFT} holds here.
\end{remark}

\begin{proof}
Without loss of generality, assume $\ell=1$, and with the support assumption of $h$, we could focus on considering $f_1(t)=B_1(t) e^{i2 \pi \phi_1(t)}$ when we analyze the associated ridge in the spectrogram and its regularity.

From Theorem \ref{Theorem Ridge regularity of STFT}, we have the control of the ridge, denoted as $\tilde{\phi}_1'(t)$.
By evaluating the phase from the ridge; that is, consider the phase of the $C^1$ function
\begin{align*}
Q(t)&:=V_{f_1}^{(h)}(t,\tilde{\phi}_1'(t))\\
&=\frac{1}{2}B_1(t)\hat{h}(\tilde{\phi}_1'(t)- \phi'_1(t)) e^{i2\pi\phi_1(t)} +C_0(t)\epsilon\,,
\end{align*}
denoted as $\tilde{\phi}_1(t)$, as the estimation of $\phi_1(t)$, we obtain
\begin{align*}
\tilde{\phi}_1(t)=\phi_1(t)+C_8(t)\epsilon\,,
\end{align*}
where $C_8(t)$ is a real valued function satisfying $|C_8(t)|\leq \frac{2|C_0(t)|}{B_1(t)}$. The amplitude can be estimated by $\tilde{B}_1(t):=\left|\frac{2Q(t)}{\hat{h}(0)}\right|$, which satisfies
\[
|\tilde{B}_1(t)-B_1(t)|\leq \frac{2|C_0(t)|}{\hat{h}(0)}\epsilon
\]
when $\epsilon$ is sufficiently small, where we use \eqref{EQ hat h' hat h hat h'' control}. Thus, we have $A_1(t)= \frac{B_1(\tilde{\phi}_1^{-1}(t))}{\tilde B_1(\tilde{\phi}_1^{-1}(t))}= 1+C_{12}(\tilde{\phi}_1^{-1}(t))\epsilon$, where $C_{12}(t)$ is a real valued function satisfying
\begin{align}\label{definition C12}
|C_{12}(t)|\leq \frac{2|C_0(t)|}{\hat{h}(0)B_1(t)} \,.
\end{align}

Since $\tilde{\phi}'_1(t)>0$, we know $\tilde{\phi}_1$ is monotonically increasing. Thus, by a direct bound with the lower bound of $\phi_1'(t)$, we have
\[
\tilde{\phi}^{-1}_1(t)=\phi^{-1}_1(t)+C_9(t)\epsilon \,,
\]
where $C_9(t)$ is a real valued function satisfying $|C_9(t)|\leq \frac{2|C_0(t)|}{HB_1(t)}$.
Denote $\theta(t)=\phi_1(\tilde{\phi}_1^{-1}(t))$. With Assumption (T3) and the lower bound of $\phi_1'(t)$, by a direct Taylor expansion we have
\[
\theta(t)=t+C_{10}(t)\epsilon\,,
\]
where $C_{10}(t)$ is a real valued function satisfying $|C_{10}(t)|\leq \frac{2|C_0(t)|}{B_1(t)}$.
With Theorem \ref{Theorem Ridge regularity of STFT}, by a direct calculation we have
\begin{align}\label{Equation theta'(t)}
\theta'(t)=\frac{\phi_1'(\tilde{\phi}_1^{-1}(t))}{\tilde{\phi}_1'(\tilde{\phi}_1^{-1}(t))}=\frac{\phi_1'(\tilde{\phi}_1^{-1}(t))}{{\phi}_1'(\tilde{\phi}_1^{-1}(t))+C_4(\tilde{\phi}_1^{-1}(t))\epsilon}
\end{align}
and
\[
\theta''(t)=\frac{\phi_1''(\tilde{\phi}_1^{-1}(t))-\tilde{\phi}_1''(\tilde{\phi}_1^{-1}(t))\frac{\phi_1'(\tilde{\phi}_1^{-1}(t))}{\tilde{\phi}_1'(\tilde{\phi}_1^{-1}(t))} }{[\tilde{\phi}_1'(\tilde{\phi}_1^{-1}(t))]^2}\,.
\]
Thus, by a direct bound, we have
\[
\theta'(t)-1=C_{11}(\tilde{\phi}_1^{-1}(t))\epsilon\,,
\]
where $C_{11}(t)$ is a real valued function satisfying $|C_{11}(t)|\leq \frac{C_4(t)}{H}$,
and
\begin{align*}
|\theta''(t)|&\leq \frac{\phi_1'(\tilde{\phi}_1^{-1}(t))+|C_7(\tilde{\phi}_1^{-1}(t))|}{[\tilde{\phi}_1'(\tilde{\phi}_1^{-1}(t))]^2}\epsilon\\
&=\,\left( \frac{\theta'(t)}{\tilde{\phi}_1'(\tilde{\phi}_1^{-1}(t))}+\frac{|C_7(\tilde{\phi}_1^{-1}(t))|}{[\tilde{\phi}_1'(\tilde{\phi}_1^{-1}(t))]^2}\right)\epsilon\\
&\leq \frac{2(H+|C_7(\tilde{\phi}_1^{-1}(t))|)}{H^2}\theta'(t)\epsilon\,,
\end{align*}
where we use \eqref{Equation theta'(t)} in the first equality and the last bound comes from the lower bound of $\phi'_1(t)$ and holds when $\epsilon$ is sufficiently small.

To finish the proof, note that by a direct calculation we have
\begin{align*}
Q'(t):=&\,-i 2 \pi [\tilde{\phi}''_1(t) V_{f_1}^{(th)}(t,\tilde{\phi}'_1(t)) - \tilde{\phi}'_1(t) V_{f_1}^{(h)}(t,\tilde{\phi}'_1(t)) ] \\
&- V_{f_1}^{(h')}(t,\tilde{\phi}'_1(t)),
\end{align*}

\noindent which by the same calculation in Lemma \ref{lemma: main taylor expansion} is bounded by
\begin{align*}
|Q'(t)|\leq &\,C_{13}(t)\epsilon\,,
\end{align*}
where
\begin{align*}
C_{13}(t)=&\,2\pi C_7(t) \left[\left(\frac{B_1(t)}{2}\hat{h}(\phi_1'(t))+|C_0(t)|\epsilon\right)t\right.\\
&\left.+\left(\frac{B_1(t)}{4\pi} \hat{h}'(\phi_1'(t))+|C_1(t)|\epsilon\right)\right]\\
&+c(1)\left(\phi_1'(t)m_1'+\frac{M}{2}m_2'\right)+3\pi|C_4(t)|B_1(t)\,.
\end{align*}
Thus, we have
\[
|\tilde{B}'_1(t)|=\frac{\left|4(Q(t)\overline{Q'(t)}+Q'(t)\overline{Q(t)})\right|}{|\hat{h}(0)|^2\tilde{B}_1(t)}\leq \frac{8C_{13}(t)}{|\hat{h}(0)|^2}\epsilon\,,
\]
where we use the fact that $|Q(t)|\leq \frac{1}{2}B_1(t)$ and $\tilde{B}_1(t)>\frac{1}{2}B_1(t)$. As a result,
\begin{align*}
A'_1(t)=&\frac{B_1'(\tilde{\phi}_1^{-1}(t))[\tilde{\phi}_1^{-1}]'(t)\tilde{B}_1(\tilde{\phi}_1^{-1}(t))}{[\tilde{B}_1(\tilde{\phi}_1^{-1}(t))]^2}\\
&-\frac{\tilde{B}_1'(\tilde{\phi}_1^{-1}(t))[\tilde{\phi}_1^{-1}]'(t){B}_1(\tilde{\phi}_1^{-1}(t)) }{[\tilde{B}_1(\tilde{\phi}_1^{-1}(t))]^2}
\end{align*}
could be controlled by
\[
|A'_1(t)|\leq C_{14}(\tilde{\phi}_1^{-1}(t))\epsilon
\]
where we use the fact that $[\tilde{\phi}_1^{-1}]'(t)\leq \frac{1}{2H}$ and
$C_{14}(t)$ is a non-negative function satisfying
\begin{align}\label{definition C14}
C_{14}(t)= \frac{1}{H{B}_1(t)} \left[ \phi_1'(t)+\frac{8C_{13}(t)}{|\hat{h}(0)|^2} \right]\,.
\end{align}
\end{proof}

Last but not least, we show that with an extra assumption, the signal $\tilde{x}^{[1]}$ also fulfill \texttt{WSFv3} so that the iteration can be carried out. Without loss of generality, assume $\ell=1$. Denote $\tilde{\phi}^{[1]}_1$ and $\tilde{B}^{[1]}_{1}$ to be the estimated phase and amplitude of $B_1(t)\cos(2\pi \phi_1(t))$.
After the warping and demodulation, we have
\[
\tilde{x}^{[1]}(t)=A^{[1]}_1(t)\cos(2\pi\theta^{[1]}_1(t))+\sum_{l=2}^\infty A^{[1]}_l(t) \cos (2\pi\theta^{[1]}_l(t) )\,,
\]
where $A^{[1]}_l(t):=\frac{B_l((\tilde{\phi}^{[1]}_1)^{-1}(t))}{\tilde{B}^{[1]}_{1}((\tilde{\phi}^{[1]}_1)^{-1}(t))}$ and $\theta^{[1]}_l(t):=\phi_l((\tilde{\phi}^{[1]}_1)^{-1}(t))$ and the deviation from the expectation that $A^{[1]}_1(t)=1$ and $\theta^{[1]}_1(t)= t$ is quantified by Theorem \ref{Theorem_warping}.
To proceed, we make an extra assumption:
\begin{enumerate}
\item[(T5)] Assume there exist $\tau_1>H$ and $\tau_2,\tau_3,\tau_4>0$ such that $\phi_\ell'(t)<\tau_1$, $B_\ell(t)>\tau_2$, $\sum_{l = 1}^{\infty} c(l) = \tau_3$ and $\sum_{l = 1}^{\infty} l B_l(t)\leq \tau_4$ for all $t\in \mathbb{R}$.
\end{enumerate}
Under this extra (T5) assumption, all functions $C_i(t)$ in Theorems \ref{Theorem Ridge regularity of STFT} and \ref{Theorem_warping} and Lemma \ref{lemma: main taylor expansion} are uniformly bounded.
As a result, we claim that $\tilde{x}^{[1]}(t)$ fulfills (T1)-(T3) with a different $\epsilon$. Indeed, it is easy to check that
$A_l\in C^1(\mathbb{R})\cap L^\infty(\mathbb{R})$ for $l=1,2,\ldots$ and for each time $t\in \mathbb{R}$. $\textup{gcd}\{l|\,A_l(t)> 0\}=1$, $A_l(t) \leq c(l)A_{\ell}(t)$, for all  $l=1,2,\ldots$ and $t\in \mathbb{R}$ for an $\ell^1$ sequence $ \{c(l)\}_{l = 1}^{\infty}$, and $\sum_{l=N+1}^\infty A_l(t)\leq \epsilon \sqrt{\sum_{l=1}^\infty A_l^2(t)}$ and $\sum_{l=N+1}^\infty lA_l(t)\leq D\sqrt{\sum_{l=1}^\infty A_l^2(t)}$ for some constant $D>0$ hold automatically.
For the regularity, since $\tilde{\phi}^{[1]}_1\in C^2(\mathbb{R})$, we have $\theta_l\in C^2(\mathbb{R})$.
The last part is quantifying the control of
\begin{align}
|\theta'_l(t) - l\theta'_1(t)|\,,\ \ |A'_{l}(t)| \ \ \mbox{and} \ \ |\theta''_{l}(t)|
\end{align}
for all $l = 1,2,\dots$. The first bound comes from a direct calculation that
\begin{align*}
|\theta'_l(t) - l\theta'_1(t)|&\,\leq |(\phi'_l(\tilde{\phi}_1^{-1}(t)) - l\theta'_1(\tilde{\phi}_1^{-1}(t)))(\tilde{\phi}_1^{-1})'(t)|\\
&\,\leq \epsilon \phi_1'(\tilde{\phi}_1^{-1}(t))(\tilde{\phi}_1^{-1})'(t)=\epsilon [\phi_1(\tilde{\phi}_1^{-1})]'(t)=\epsilon \theta_1'(t)\,.
\end{align*}
The second bound is
\begin{align*}
|A'_{l}(t)|&\,= \left|\frac{B_l'(\tilde{\phi}_1^{-1}(t))[\tilde{\phi}_1^{-1}]'(t)\tilde{B}_1(\tilde{\phi}_1^{-1}(t))}{[\tilde{B}_1(\tilde{\phi}_1^{-1}(t))]^2}\right.\\
&\left.-\frac{\tilde{B}_1'(\tilde{\phi}_1^{-1}(t))[\tilde{\phi}_1^{-1}]'(t){B}_l(\tilde{\phi}_1^{-1}(t)) }{[\tilde{B}_1(\tilde{\phi}_1^{-1}(t))]^2}\right|\\
&\,\leq \left[\frac{2}{\tau_2}+\frac{16\|C_{13}\|_\infty}{|\hat{h}(0)|^2H\tau_2}\right]\epsilon c(l)\theta_1'(t)\,.
\end{align*}
The third bound also comes from a direct calculation
\begin{equation}\label{bound of theta''_l}
\begin{aligned}
|\theta_l''(t)|&\,=\phi_l''(\tilde{\phi}_1^{-1}(t))[(\tilde{\phi}_1^{-1})'(t)]^2\\
&\quad +\phi_l'(\tilde{\phi}_1^{-1}(t))(\tilde{\phi}_1^{-1})''(t)\\
&\,\leq \epsilon l\theta_1'(t)(\tilde{\phi}_1^{-1})'(t)+(l\phi_1'(\tilde{\phi}_1^{-1}(t))\\
&\quad+\epsilon\phi_1'(t))\frac{\tilde{\phi}_1''(\tilde{\phi}_1^{-1}(t))(\tilde{\phi}_1^{-1})'(t)}{|\tilde{\phi}_1'(\tilde{\phi}_1^{-1}(t))|^2}\\
&\,\leq \left[ \frac{1}{H}+\frac{2\|C_7\|_\infty}{H^2}\right]\epsilon l\theta_1'(t) \,.
\end{aligned}
\end{equation}
Moreover, we know from \eqref{bound of theta''_l} that $\sup_{l;\,B_l\neq 0}\|\theta''_l\|_\infty=O(\epsilon)$.
Thus, $\tilde{x}^{[1]}(t)$ fulfills (T1)-(T3) with $\epsilon^{[1]}:=C^{[1]}\epsilon$, where
\[
C^{[1]}=\max\left\{\left[\frac{2}{\tau_2}+\frac{16\|C_{13}\|_\infty}{|\hat{h}(0)|^2H\tau_2}\right],\left[ \frac{1}{H}+\frac{2\|C_7\|_\infty}{H^2}\right],1\right\}\,.
\]
Therefore, Theorems \ref{Theorem Ridge regularity of STFT} and \ref{Theorem_warping} and Lemma \ref{lemma: main taylor expansion} can be directly applied to $\tilde{x}^{[1]}(t)$, and hence we can again apply the same warping and demodulation process to $\tilde{x}^{[1]}(t)$ with guarantees.
We shall mention that the signal property change is most dramatic upon the first iteration. Indeed, after the first iteration, the AM of the dominant harmonic becomes close to $1$, and the phase becomes close to $t$ in each iteration, which are ``very different'' from  the dominant harmonic of $x$.
Note that (T1)-(T5) hold with different constants for $\tilde{x}^{[2]}(t)$ with $\epsilon^{[2]}$. Among all constants, $\tau_3=\sum_{l = 1}^{\infty} c(l)$ is fixed, but $H$, $\tau_1$ and $\tau_2$ become close to $1$, since $(\theta^{[1]}_1)'(t)\approx 1$ and $A^{[1]}_1(t)\approx 1$ and $\tau_4$ is still bounded. Note that the kernel, and hence the absolute moments, used will also be different since it depends on $H$.
So, Theorems \ref{Theorem Ridge regularity of STFT} and \ref{Theorem_warping} and Lemma \ref{lemma: main taylor expansion} can be directly applied to $\tilde{x}^{[2]}(t)$, and hence $\tilde{x}^{[k]}(t)$, $k=3,4,\ldots$.

\section{Discussion and conclusion}\label{section:discussion}

Motivated by challenges from analyzing real world signals, we propose a new model, \texttt{WSFv3}, to capture sudden WSF changes in a nonstationary oscillatory signal. The sudden WSF change is modeled as a change point in the amplitude or phase of some non-dominant harmonics. We propose an iterative warping and clustering algorithm to estimate all different WSFs and hence change points of WSFs from the signal.

The main features of the proposed algorithm could be observed in its application to ABP and ECG. As we can see from Fig. \ref{fig:ABP}, one step of warping and demodulation is not sufficient for the ABP signal. This is due to the inevitable error in the phase estimation, and an iterative strategy proved to be useful to alleviate this problem. On the other hand, the synchronization step by GCL might be needed, especially when high harmonic IFs are not exact multiples of the fundamental component IF. This step improves the clustering performance, and leads to clusters that are more concentrated and reflect better the actual WSFs present in the signal.
The analyzed ECG signal shows us an example where warping with the fundamental component might not work well. Although we observe a component of constant warped frequency equal to 1 if we use the fundamental component to do the warping process, the high harmonics possess fluctuating IFs. Although these variations are bounded, they seem to be the cause of the difficulties of isolating WSFs with physiological meanings. The strategy of warping with the dominant harmonic (the second one in this ECG case) proved to lead to a satisfactory result, which is due to the stronger SNR.

Note that our algorithm is related to, but different from, the traditional signal segmentation approaches, where the algorithms are mainly based on the identified fiducial points (or landmarks). In general, landmark search might be difficult, especially in the presence of noise, and it depends on the signal characteristics.
Our proposed approach is however via warping the signal so that the landmark search is no longer needed. Thus it is more universal.
We shall also mention that when the signal is segmented in the traditional way by taking the landmarks into account, the segment length usually varies from one to another due to the time-varying frequency. Thus, if the purpose is extracting the dynamics hidden in the time-varying WSFs, the mission might be challenged since we cannot directly compare segments of different lengths with the traditional Euclidean metric. One solution is truncating the segments so that all truncated segments have the same length \cite{YTLin2019wave,wang2020novel}. While this approach works, the captured information depends not only on time-varying  WSFs but also on IF. With our warping approach, the impact of IF is gone. Its application to extract dynamics hidden in the WSFs will be explored in our future work.
To our knowledge, our algorithm is the first change point detection algorithm on the level of WSFs.

We shall clarify a potential confusing point. At the first glance, the signals of phase shifting key (PSK) and frequency shifting key (FSK) in communication systems might be modeled by \texttt{WSFv3}. For both cases, the signal is sinusoidal with constant amplitude, but the phase is either non-continuous (in PSK) or non-smooth (in FSK). This fact is contrary to condition T2, which states that the phase of the dominant harmonic must be twice differentiable. In other words, if we view a signal of PSK or FSK as an oscillatory signal with the cosine function as its WSF, its dominant harmonic is the signal itself, and the non-continuous or non-smooth phase assumption of PSK or FSK does not satisfy the model. To detect the phase change point of this kind of signal is actually a different challenging problem, and we refer readers with interest to \cite{zhou2020frequency} for a recent development in that direction.

This study has limitations. As the first algorithm in the field, a statistical inference framework needs to be further developed. The dynamic of WSFs itself contains rich biomedical information \cite{YTLin2019wave,wang2020novel}, and how to incorporate it into our current analysis framework is not studied in this paper. We assume that there is one oscillatory component in the signal and focus on detecting change points of WSF. When there are more than one oscillatory components, like the trans-abdominal maternal ECG \cite{su2017extract} or PPG \cite{li2019non,ShelleyAlianWu2022,maity2022ppgmotion}, we need to decompose the signal into separate components before applying our algorithm. These interesting topics are however out of the scope of this paper, and will be reported in the future work.

\section*{Acknowledgement}
The authors would like to thank Professor James Nolen for valuable discussions on the topic and the anonymous reviewers for their constructive comments.


\end{document}